\documentclass{article} 
\usepackage{iclr2027_conference,times}

\usepackage{amsmath,amsfonts,bm}

\def\eqref#1{equation~\ref{#1}}

\def\1{\bm{1}}

\DeclareMathAlphabet{\mathsfit}{\encodingdefault}{\sfdefault}{m}{sl}
\SetMathAlphabet{\mathsfit}{bold}{\encodingdefault}{\sfdefault}{bx}{n}

\newcommand{\E}{\mathbb{E}}

\usepackage[utf8]{inputenc}
\usepackage[T1]{fontenc}
\usepackage{amsmath,amssymb,amsfonts}
\usepackage{booktabs}
\usepackage{array}
\usepackage{multirow}
\usepackage{pifont}
\usepackage{graphicx}
\usepackage{xcolor}
\usepackage{tikz}
\usetikzlibrary{arrows.meta,positioning,calc,decorations.pathreplacing}
\usepackage{url}
\usepackage{hyperref}
\usepackage{amsthm}
\newtheorem{theorem}{Theorem}[section]

\newtheorem{lemma}[theorem]{Lemma}
\theoremstyle{remark}
\newtheorem{remark}[theorem]{Remark}

\newcommand{\method}{\texorpdfstring{\texttt{xPRT}}{xPRT}}   

\newcommand{\Lxp}{\mathcal{L}_{\text{xp}}}          
\newcommand{\Rcross}{\mathcal{L}_{\text{cross}}}    
\newcommand{\Rself}{\mathcal{L}_{\text{self}}}      
\newcommand{\Lcross}{\Rcross}
\newcommand{\Lself}{\Rself}
\newcommand{\Lmlm}{\mathcal{L}_{\text{MLM}}}
\newcommand{\Lunif}{\mathcal{L}_{\text{unif}}}
\newcommand{\Lret}{\mathcal{L}_{\text{ret}}}
\newcommand{\sg}{\operatorname{sg}}

\newcommand{\xPRT}{\texttt{xPRT}}

\title{\xPRT: Cross-Passage Retrieval Pretraining}
\author{Danush Khanna\thanks{This work was done during an internship at Microsoft Research.} , Himanshu Singhal, Pranjal A. Chitale \& Bishal Santra \\
Microsoft Research \\
\texttt{danush.research@gmail.com} \\
\texttt{\{t-hsinghal,pchitale,bishalsantra\}@microsoft.com}
}

\iclrfinalcopy 
\begin{document}

\maketitle
\lhead{Preprint} 

\begin{abstract}
Dense retrievers are a foundational component underlying modern RAG and search systems. Training these models typically requires a curated collection of annotated query-document pairs, which are prohibitively expensive to obtain outside a handful of well-resourced domains. This has motivated a line of work that leverages large-scale unsupervised retrieval pretraining from unlabelled text alone. Most such methods learn by optimizing a language modelling or contrastive learning objective over\textit{ short spans within a document}. This raises the question: \textit{how can we extract richer learning signals from a single document for retrieval pretraining?} To address it, we introduce \textbf{\xPRT} (Cross-Passage Retrieval Pretraining), an unsupervised framework that learns from the key information distributed across an entire document and the connections between its parts. We evaluate \xPRT\ against five state-of-the-art objectives on MTEB English v2, continuing pretraining from a public unsupervised checkpoint while holding the corpus, schedule, and training procedure fixed. Across diverse tasks spanning all task types, we observe that \xPRT\ consistently outperforms the strongest baseline by $\mathbf{3.2\%}$ on an average. Notably, in this continual pretraining setting, many baseline methods incur substantial performance regressions, whereas our method delivers the largest complementary gains over the initial pretrained checkpoint. Controlled ablations confirm that \xPRT's performance gains stem from its proposed cross-passage signals. \xPRT\ demonstrates that a single unlabelled document is a much richer source of retrieval signal than prior work assumed, and offers \textbf{complementary gains} to what current retrievers already capture.
\end{abstract}

\section{Introduction}
\label{sec:intro}

Dense retrievers are a core component of modern search, open-domain question answering and
retrieval-augmented generation (RAG). A dual encoder maps queries and passages to vectors in a
shared space and retrieves candidates from a prebuilt index by approximate nearest-neighbour
search \citep{karpukhin2020dpr, reimers2019sbert, zhao2024survey}. In RAG, a language model
answers from the passages it is given \citep{lewis2020rag}, so a passage the retriever misses
cannot be recovered later in the pipeline. Since the encoder is the whole retriever, retrieval
quality rests on how well a single pooled vector stands in for a passage. Learning that vector
has traditionally relied on costly supervision: large collections of annotated query-passage pairs, mined hard
negatives and cross-encoder filtering \citep{bajaj2016msmarco, xiong2021ance, qu2021rocketqa}.
Such data exist for a handful of well-resourced domains and are prohibitively expensive to curate
elsewhere.

This cost motivates learning from unlabelled text. Pre-trained language models such as BERT
\citep{devlin2019bert}, RoBERTa \citep{liu2019roberta} and T5 \citep{raffel2020t5} are the
standard initialization for dense retrievers \citep{karpukhin2020dpr, xiong2021ance}. However, they are
pre-trained with token-level supervision, and the
pooled vector never learns to represent a passage as a whole, in a manner that is suitable for retrieval
\citep{chang2020pretraining, reimers2019sbert}, and thus labelled data must close the gap.

\emph{Retrieval-oriented pre-training} aims to close part of this gap without labels, by training the
pooled vector directly on unlabelled text.
\textbf{Contrastive} methods build pseudo-relevant pairs from the corpus, such as an inverse-cloze
pair or two crops of the same document, and train the pooled vector to prefer its counterpart over
passages from other documents \citep{lee2019ict, izacard2022contriever, ram2022spider,
gao2022cocondenser}. \textbf{Bottleneck auto-encoding} methods instead train the pooled vector to
summarize its passage \citep{gao2021condenser, xiao2022retromae, wang2023simlm,
ma2024bowdpr, li2023challenging, shen2023lexmae}. The encoder compresses a passage into the vector, and a decoder must reconstruct
the passage from it, so the vector is pushed to retain the passage's content. 
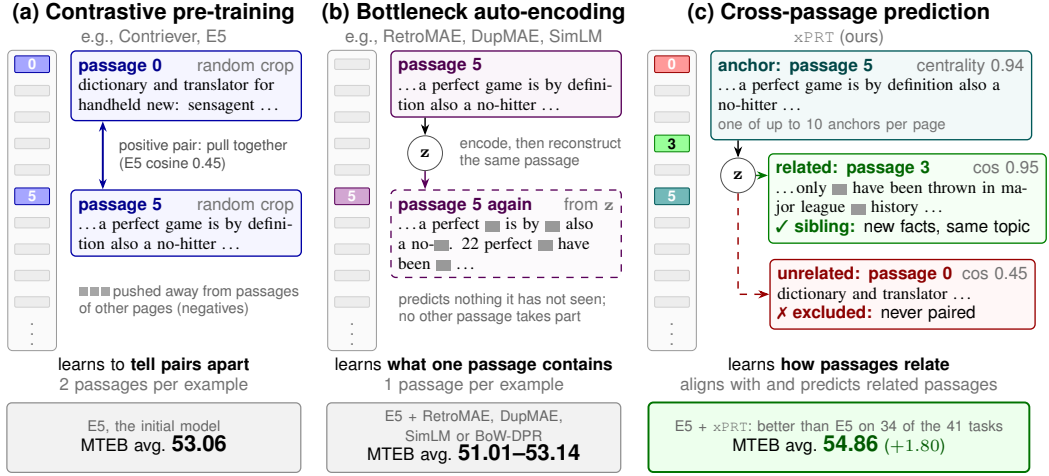
\begin{figure}[!t]
  \centering
  \resizebox{\linewidth}{!}{
\providecommand{\msk}{\textcolor{black!40}{\rule[-0.2ex]{0.85em}{1.3ex}}}
\begin{tikzpicture}[
  font=\sffamily\scriptsize,
  bar/.style={draw=black!18, fill=black!7, rounded corners=0.6pt, minimum width=0.5cm, minimum height=0.15cm, inner sep=0pt},
  hb/.style={bar, minimum height=0.24cm, font=\sffamily\tiny\bfseries},
  txt/.style={rounded corners=2pt, align=left, inner sep=3pt, line width=0.55pt, anchor=north west},
  vec/.style={draw=black!70, fill=white, circle, minimum size=0.5cm, inner sep=0pt},
  arr/.style={-{Stealth[length=1.6mm]}, semithick},
  ttl/.style={font=\sffamily\footnotesize\bfseries, anchor=north},
  sub/.style={font=\sffamily\scriptsize, text=black!55, anchor=north},
  lrn/.style={align=center, anchor=center},
  res/.style={rounded corners=2pt, align=center, inner sep=2.5pt, line width=0.5pt, anchor=south, minimum height=1.0cm},
  note/.style={font=\sffamily\tiny, text=black!60, align=left},
  page/.style={black!40, fill=black!2, rounded corners=2pt},
]

\begin{scope}
  \node[ttl] at (2.15,6.9) {(a) Contrastive pre-training};
  \node[sub] at (2.15,6.52) {e.g., Contriever, E5};
  \draw[page] (0.05,1.78) rectangle (0.75,6.1);
  \foreach \i in {0,...,9} \node[bar] (pa\i) at (0.4,{5.9-0.38*\i}) {};
  \node[text=black!50] at (0.4,2.12) {$\vdots$};
  \node[hb, draw=blue!65!black, fill=blue!35, text=white] at (pa0) {0};
  \node[hb, draw=blue!65!black, fill=blue!35, text=white] at (pa5) {5};

  \node[txt, draw=blue!65!black, fill=blue!4, text width=3.1cm] (ta0) at (0.95,6.1) {%
    {\bfseries\color{blue!55!black} passage 0}\hfill{\color{black!55} random crop}\\
    {\rmfamily dictionary and translator for handheld new: sensagent \ldots}};
  \node[txt, draw=blue!65!black, fill=blue!4, text width=3.1cm] (ta5) at (0.95,4.1) {%
    {\bfseries\color{blue!55!black} passage 5}\hfill{\color{black!55} random crop}\\
    {\rmfamily \ldots a perfect game is by definition also a no-hitter \ldots}};
  \draw[{Stealth[length=1.5mm]}-{Stealth[length=1.5mm]}, thick, blue!60!black]
    ($(ta0.south west)!0.14!(ta0.south east)$) -- ($(ta5.north west)!0.14!(ta5.north east)$)
    node[midway, right=3pt, note, text=black!70] {positive pair: pull together\\(E5 cosine 0.45)};
  \node[note, anchor=north west] at (0.95,2.85) {{\color{black!35}\rule{1.1ex}{1.1ex}\,\rule{1.1ex}{1.1ex}\,\rule{1.1ex}{1.1ex}}\ pushed away from passages\\of other pages (negatives)};

  \node[lrn, text width=4.1cm] at (2.15,1.42) {learns to \textbf{tell pairs apart}\\{\color{black!55}2 passages per example}};
  \node[res, draw=black!45, fill=black!5, text width=4.05cm] at (2.15,0.02) {{\tiny\color{black!60}E5, the initial model}\\MTEB avg.\ {\small\bfseries 53.06}};
\end{scope}

\begin{scope}[xshift=4.6cm]
  \node[ttl] at (2.15,6.9) {(b) Bottleneck auto-encoding};
  \node[sub] at (2.15,6.52) {e.g., RetroMAE, DupMAE, SimLM};
  \draw[page] (0.05,1.78) rectangle (0.75,6.1);
  \foreach \i in {0,...,9} \node[bar] (pb\i) at (0.4,{5.9-0.38*\i}) {};
  \node[text=black!50] at (0.4,2.12) {$\vdots$};
  \node[hb, draw=violet!70!black, fill=violet!35, text=white] at (pb5) {5};

  \node[txt, draw=violet!70!black, fill=violet!4, text width=3.1cm] (tb5) at (0.95,6.1) {%
    {\bfseries\color{violet!60!black} passage 5}\\
    {\rmfamily \ldots a perfect game is by definition also a no-hitter \ldots}};
  \node[vec] (zb) at (1.45,4.62) {$\mathbf{z}$};
  \draw[arr] (tb5.south -| zb) -- (zb);
  \node[note, anchor=west] at ($(zb.east)+(0.12,0)$) {encode, then reconstruct\\the same passage};
  \node[txt, draw=violet!70!black, dashed, fill=violet!2, text width=3.1cm] (tb5r) at (0.95,4.1) {%
    {\bfseries\color{violet!60!black} passage 5 again}\hfill{\color{black!55} from $\mathbf{z}$}\\
    {\rmfamily \ldots a perfect \msk{} is by \msk{} also a no-\msk{}. 22 perfect \msk{} have been \msk{} \ldots}};
  \draw[arr, violet!70!black] (zb) -- (tb5r.north -| zb);
  \node[note, anchor=north west] at (0.95,2.72) {predicts nothing it has not seen;\\no other passage takes part};

  \node[lrn, text width=4.1cm] at (2.15,1.42) {learns \textbf{what one passage contains}\\{\color{black!55}1 passage per example}};
  \node[res, draw=black!45, fill=black!5, text width=4.05cm] at (2.15,0.02) {{\tiny\color{black!60}E5 + RetroMAE, DupMAE,\\SimLM or BoW-DPR}\\MTEB avg.\ {\small\bfseries 51.01--53.14}};
\end{scope}

\begin{scope}[xshift=9.2cm]
  \node[ttl] at (2.8,6.9) {(c) Cross-passage prediction};
  \node[sub] at (2.8,6.52) {\method{} (ours)};
  \draw[page] (0.05,1.78) rectangle (0.75,6.1);
  \foreach \i in {0,...,9} \node[bar] (pc\i) at (0.4,{5.9-0.38*\i}) {};
  \node[text=black!50] at (0.4,2.12) {$\vdots$};
  \node[hb, draw=red!60!black, fill=red!40, text=white] at (pc0) {0};
  \node[hb, draw=green!50!black, fill=green!40] at (pc3) {3};
  \node[hb, draw=teal!70!black, fill=teal!55, text=white] at (pc5) {5};

  \node[txt, draw=teal!70!black, fill=teal!5, text width=4.4cm] (tc5) at (0.95,6.1) {%
    {\bfseries\color{teal!55!black} anchor: passage 5}\hfill{\color{black!55} centrality 0.94}\\
    {\rmfamily \ldots a perfect game is by definition also a no-hitter \ldots}\\
    {\tiny\color{black!55} one of up to 10 anchors per page}};
  \node[vec] (zc) at (1.35,4.3) {$\mathbf{z}$};
  \draw[arr] (tc5.south -| zc) -- (zc);
  \node[txt, draw=green!50!black, fill=green!4, text width=3.78cm] (tc3) at (1.78,4.62) {%
    {\bfseries\color{green!40!black} related: passage 3}\hfill{\color{black!55} cos 0.95}\\
    {\rmfamily \ldots only \msk{} have been thrown in major league \msk{} history \ldots}\\
    {\color{green!40!black}\bfseries\ding{51}\ sibling:} new facts, same topic};
  \node[txt, draw=red!60!black, fill=red!3, text width=3.6cm] (tc0) at (1.8,3.1) {%
    {\bfseries\color{red!55!black} unrelated: passage 0}\hfill{\color{black!55} cos 0.45}\\
    {\rmfamily dictionary and translator \ldots}\\
    {\color{red!55!black}\bfseries\ding{55}\ excluded:} never paired};
  \draw[arr, green!45!black] (zc) -- (zc -| tc3.west);
  \draw[arr, red!60!black, dashed] (zc.south) |- (tc0.west);

  \node[lrn, text width=5.4cm] at (2.8,1.42) {learns \textbf{how passages relate}\\{\color{black!55}aligns with and predicts related passages}};
  \node[res, draw=green!45!black, fill=green!7, line width=0.8pt, text width=5.3cm] at (2.8,0.02) {{\tiny\color{black!60}E5 + \method{}: better than E5 on 34 of the 41 tasks}\\MTEB avg.\ {\small\bfseries\color{green!35!black} 54.86}\ {\color{green!35!black}($+1.80$)}};
\end{scope}
\end{tikzpicture}}

\caption{\textbf{How pre-training objectives use a document's passages}, on one FineWeb page.
  (a) Contrastive pre-training separates a passage pair from other pages'; a random pair can wrongly link an article to a dictionary widget.
  (b) Bottleneck auto-encoding reconstructs a passage from its own embedding $\mathbf{z}$, so it need not anticipate unseen content; 
  (c) \xPRT\ makes an anchor's embedding predict its four most related passages, adding new facts on the query--answer relation a retriever needs; unrelated passages (the widget) are neither anchors nor targets.
  Bottom: average over 41 MTEB(eng, v2) tasks for E5 and E5 after continued pre-training on 100K documents (Table~\ref{tab:main}).
  }
  \label{fig:motivation}
\end{figure}
Both families, however, make some strong assumptions for sourcing the learning signal that might contradict with what a good retriever needs. Auto-encoders reconstruct a passage
from itself, so the pooled vector is never asked to anticipate content it has not read, which is
exactly what a query asks of the passage that answers it. Contrastive methods look beyond the
span, but pair it with a second span chosen by position or at random. They reduce the relation
between the two to a single label, same document or not, applied with equal weight whether the
spans share a subject or only a URL. On web text, a navigation block is as likely to be drawn as
the paragraph that states the topic \citep{penedo2024fineweb}, and nothing else in the document
answers a query drawn from boilerplate. CoT-MAE \citep{wu2023cotmae} does look outside the
passage, but supervises only a single adjacent context span. Two document-level signals thus go
unused: \emph{which} passages carry a document's subject, and \emph{how} those passages relate to
one another. This raises our central question: \emph{how much more retrieval signal can a single
unlabelled document supply?}

We propose \xPRT\ (Cross-Passage Retrieval Pre-training), a self-supervised objective that
exploits both signals using the document alone. Content, not position, decides what a training
example contains. A passage becomes an \emph{anchor} only if its embedding under the frozen
initial encoder lies close to the centroid of its document, which keeps boilerplate out of the
query side. The anchor's \emph{siblings} are the passages of the same document most related to it
under fused dense and lexical similarity (\S\ref{sec:selection}). Supervision then runs
\emph{across} passages rather than within one. A contrastive loss makes the anchor single out its
most related sibling among passages of other documents (\S\ref{sec:xp}). A lightweight decoder that
sees the anchor only through its embedding predicts the masked tokens of sibling passages the encoder
never read (\S\ref{sec:pred}). This mirrors the relation retrieval needs between a query and the
passage that answers it. 

We continue pre-training the public \texttt{e5-base-unsupervised} checkpoint \citep{wang2022e5} on
FineWeb \citep{penedo2024fineweb} under each objective. We hold the backbone, corpus, schedule,
batch size, passage length, retention regularizer and seeds fixed, and evaluate every checkpoint
under one protocol without fine-tuning, so that the objective is the only variable
(\S\ref{sec:results}). On MTEB(eng, v2) \citep{muennighoff2023mteb, enevoldsen2025mmteb}, \xPRT\
outperforms the strongest of five baselines across three data scales experimented with, by $\mathbf{3.2\%}$ on average over 41 tasks and by
$\mathbf{6.1\%}$ nDCG@10 on the ten retrieval tasks.

Our contributions are:

\begin{itemize}
\item We identify two document-level signals that unsupervised retrieval objectives leave unused:
which passages represent a document, and how its passages relate. We propose \xPRT, which learns
from both by making the embedding of a representative passage predict the related passages it
never read (\S\ref{sec:method}).

\item In a controlled continual pre-training comparison against five established objectives,
\xPRT\ is the only one that improves its starting checkpoint substantially. It gains 1.8 to 1.9
points on MTEB(eng, v2), on every task type, at every corpus scale from 100K to 1M documents. No
baseline gains more than 0.4 points at any scale, and four of the five end below the checkpoint
at 100K (\S\ref{sec:main} and \S\ref{sec:breadth}).

\item Controlled ablations trace the gain to the cross-passage signal from curated related
passages. Removing cross-passage prediction and pairing each anchor with a random passage of its
document costs $9.3$ nDCG@10. Within-passage even given as much text, or up to five times the
steps, stays between $18.8$ - $19.2$ nDCG@10, below the strongest baseline (\S\ref{sec:ablation}).
\end{itemize}





\section{Cross-Passage Retrieval Pre-Training}\label{sec:method}

Established retrieval pre-training objectives learn a passage's embedding from the passage itself or from heuristic pairs, such as random crops or adjacent spans. \method{} instead learns from \emph{cross-passage signals}: the embedding of a passage that represents its document must single out, and predict the content of, the passages of the same document that are most related to it or complements it, e.g. answer a question posed by the first passage. The encoder $f_\theta$, a 12-layer Transformer, is later deployed as the retriever. Documents are split into passages of 128 tokens, and each training example is a \emph{group} $g=(a,s_1,\dots,s_n)$: an anchor passage $a$ and $n=4$ sibling passages of the same document, ranked by their relatedness to the anchor (\S\ref{sec:selection}). For a passage $x$, $e_x$ is the mean of the encoder's final hidden states, the embedding used for retrieval. During training, the encoder reads a light view $m_L(x)$ of every passage, with 15\% of the tokens replaced by [MASK], and a one-layer decoder $h_\phi$, discarded after pre-training, reads a heavy view $m_H(x)$, with 60\% masked at the positions $M_H(x)$, a superset of the light mask (Appendix~\ref{app:impl}). The objective combines two cross-passage terms, which carry this signal, with a self-prediction term and two standard regularizers:

\begin{equation}
  \mathcal{L}_{\text{\method{}}}
  = \underbrace{\Lxp}_{\text{contrastive xp signal (\S\ref{sec:xp})}}
  + \underbrace{\lambda_c\Lcross + \Lself}_{\text{content supervision signal(\S\ref{sec:pred})}}
  + \underbrace{\Lmlm + \lambda_u\Lunif}_{\text{regularizers (\S\ref{sec:pred})}} ,
  \label{eq:objective}
\end{equation}
with $\lambda_c = 0.1$, $\lambda_u = 0.01$, and all other weights 1 (Table ~\ref{tab:hparams}). Figure~\ref{fig:method} shows how terms are computed for a group.
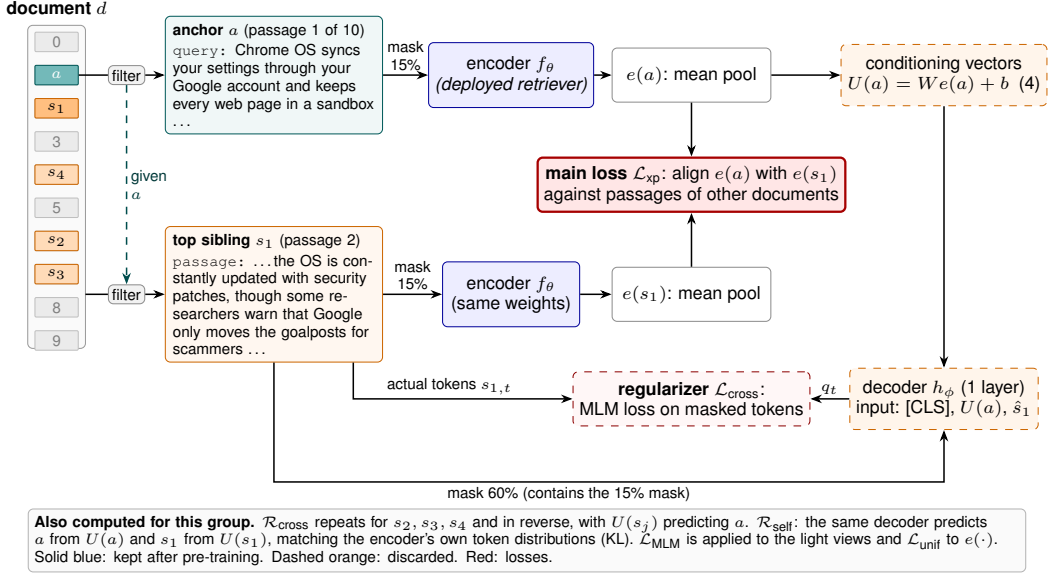
\begin{figure}[t]

  \centering
  \resizebox{\linewidth}{!}{\resizebox{\linewidth}{!}{%
\begin{tikzpicture}[
  font=\sffamily\scriptsize,
  box/.style={draw, rounded corners=2pt, align=center, inner sep=2.5pt},
  txt/.style={box, text width=2.8cm, align=left, font=\sffamily\tiny, inner sep=3pt},
  enc/.style={box, draw=blue!60!black, fill=blue!7, minimum width=1.9cm, minimum height=0.9cm},
  aux/.style={box, draw=orange!80!black, fill=orange!8, dashed, minimum width=2.5cm, minimum height=0.85cm},
  outbox/.style={box, draw=black!45, fill=white, minimum width=2.2cm, minimum height=0.75cm},
  mainloss/.style={box, draw=red!65!black, fill=red!10, line width=0.9pt, minimum height=0.75cm},
  reg/.style={box, draw=red!55!black, fill=red!3, dashed, minimum height=0.75cm},
  filt/.style={box, draw=black!45, fill=black!5, font=\sffamily\tiny, inner sep=1.5pt},
  arr/.style={-{Stealth[length=1.6mm]}, semithick},
  lbl/.style={font=\sffamily\tiny, inner sep=1.5pt},
  bar/.style={draw=black!25, fill=black!8, rounded corners=0.5pt, minimum width=0.6cm, minimum height=0.27cm, inner sep=0pt, font=\sffamily\tiny},
]
\node[font=\sffamily\scriptsize\bfseries] at (0.45,3.22) {document $d$};
\draw[black!40, fill=black!2, rounded corners=2pt] (0.05,-1.5) rectangle (0.85,2.98);
\foreach \i in {0,...,9} \node[bar] (d\i) at (0.45,{2.75-0.46*\i}) {\color{black!50}\i};
\node[bar, draw=teal!70!black, fill=teal!55, text=white] at (d1) {\bfseries $a$};
\node[bar, draw=orange!80!black, fill=orange!50] at (d2) {$s_1$};
\node[bar, draw=orange!80!black, fill=orange!28] at (d6) {$s_2$};
\node[bar, draw=orange!80!black, fill=orange!28] at (d7) {$s_3$};
\node[bar, draw=orange!80!black, fill=orange!28] at (d4) {$s_4$};
\node[filt] (fa) at (1.42,2.29) {filter};
\node[filt] (fs) at (1.42,-0.75) {filter};
\draw[semithick] (d1.east) -- (fa.west);
\draw[semithick] (0.85,-0.75) -- (fs.west);
\draw[arr, dashed, teal!60!black] (fa.south) -- node[right, lbl, text=teal!45!black, align=left] {given\\ $a$} (fs.north);
\node[txt, draw=teal!70!black, fill=teal!6, anchor=west] (a) at (1.95,2.29)
  {\textbf{anchor $a$} (passage 1 of 10)\\[1pt] \texttt{query:} Chrome OS syncs your settings through your Google account and keeps every web page in a sandbox \dots};
\node[txt, draw=orange!80!black, fill=orange!6, anchor=west] (s) at (1.95,-0.75)
  {\textbf{top sibling $s_1$} (passage 2)\\[1pt] \texttt{passage:} \dots the OS is constantly updated with security patches, though some researchers warn that Google only moves the goalposts for scammers \dots};
\draw[arr] (fa.east) -- (a.west);
\draw[arr] (fs.east) -- (s.west);
\node[enc] (ea) at (6.75,2.29) {encoder $f_\theta$\\ \textit{(deployed retriever)}};
\node[enc] (es) at (6.75,-0.75) {encoder $f_\theta$\\ (same weights)};
\draw[arr] (a.east) -- node[above, lbl, align=center] {mask\\ 15\%} (ea.west);
\draw[arr] (s.east) -- node[above, lbl, align=center] {mask\\ 15\%} (es.west);
\node[outbox] (oa) at (9.25,2.29) {$e(a)$: mean pool};
\node[outbox] (os) at (9.25,-0.75) {$e(s_1)$: mean pool};
\draw[arr] (ea) -- (oa);
\draw[arr] (es) -- (os);
\node[mainloss] (Lxp) at (9.25,0.77) {\textbf{main loss} $\mathcal{L}_{\text{xp}}$: align $e(a)$ with $e(s_1)$\\ against passages of other documents};
\draw[arr] (oa) -- (Lxp);
\draw[arr] (os) -- (Lxp);
\node[aux] (U) at (12.75,2.29) {conditioning vectors\\ $U(a)=We(a)+b$ \ (4)};
\node[aux] (dec) at (12.75,-2.2) {decoder $h_\phi$ (1 layer)\\ input: [CLS], $U(a)$, $\hat{s}_1$};
\node[reg] (R) at (9.25,-2.2) {\textbf{regularizer} $\mathcal{L}_{\text{cross}}$:\\ MLM loss on masked tokens};
\draw[arr] (oa) -- (U);
\draw[arr] (U) -- (dec);
\draw[arr] (dec) -- node[above, lbl] {$q_t$} (R);
\draw[arr] ($(s.south)+(1.1,0)$) |- node[pos=0.72, above, lbl] {actual tokens $s_{1,t}$} (R.west);
\draw[arr] (s.south) |- node[pos=0.72, below, lbl] {mask 60\% (contains the 15\% mask)} ($(dec.south)+(0,-0.72)$) -- (dec.south);
\node[box, draw=black!30, fill=black!2, text width=13.65cm, align=left, font=\sffamily\tiny, anchor=north west] at (0.05,-3.72)
  {\textbf{Also computed for this group.}
   $\mathcal{R}_{\text{cross}}$ repeats for $s_2,s_3,s_4$ and in reverse, with $U(s_j)$ predicting $a$.
   $\mathcal{R}_{\text{self}}$: the same decoder predicts $a$ from $U(a)$ and $s_1$ from $U(s_1)$, matching the encoder's own token distributions (KL).
   $\mathcal{L}_{\text{MLM}}$ is applied to the light views and $\mathcal{L}_{\text{unif}}$ to $e(\cdot)$.
   Solid blue: kept after pre-training. Dashed orange: discarded. Red: losses.};
\end{tikzpicture}%
}}
  \caption{\textbf{\method{} training.} \emph{Left:} From a document, first a group of anchor and n sibling crops are selected. Anchors are passages whose embedding is close to the embedding of the full document; the siblings of an anchor are the four other passages most related to it by a fused dense and TF-IDF similarity metric, so the anchor determines its siblings. \emph{Right:} the encoder, the only model we keep, encodes the anchor/sibling passages with light (~15\%) masking. The cross-passage loss $\Lxp$ aligns the embeddings of the anchor and its top sibling against anchor/sibling passages of other documents. The content supervision loss $\Rcross$ asks a one-layer decoder, which receives the anchor's embedding as four conditioning vectors $U(a)$, to recover tokens from its heavily masked siblings. Excerpts are paraphrased from a real group of a 10-passage document.}
  \label{fig:method}
\end{figure}

\subsection{Selecting related passages}\label{sec:selection}

A sibling is useful only if it is related to its anchor, and an anchor only if it represents its document rather than, say, navigation text or a widget (Figure~\ref{fig:motivation}). We select groups once, before training and for both training regimes, with frozen E5 \citep{wang2022e5} as a dense scorer and inverse document frequency (IDF) as a lexical one \citep{sparckjones1972idf}. An anchor must be \emph{on-topic}: the cosine similarity between its embedding and the centroid of its document's passage embeddings must exceed a threshold $\tau_a$ (Appendix~\ref{app:selection}). Among the eligible passages, we keep up to ten anchors per document, preferring those that are both central and lexically informative, as combined by reciprocal rank fusion (RRF; \citealp{cormack2009rrf}). For each anchor, every other passage of the document is a candidate sibling; we rank the candidates by fusing their dense and TF-IDF similarity to the anchor, again with RRF, and keep the top four. Siblings therefore come from anywhere in the document rather than from adjacent spans \citep{wu2023cotmae}, and every document contributes several groups: on the 100K corpus, 99\% of documents yield groups (7.0 per document on average), which cover 68\% of all passages.

\subsection{Cross-Passage Contrastive Signals}\label{sec:xp}

The main objective, $\Lxp$, asks the embedding of an anchor to single out its most related sibling among passages of other documents. A symmetric InfoNCE loss \citep{oord2018cpc} with cosine similarity $\cos$ and temperature $\tau = 0.05$ aligns each anchor with its top-ranked sibling $s_1$:
\begin{equation}
  \Lxp = -\frac12\left[
  \log\frac{\exp\big(\cos(e_a,e_{s_1})/\tau\big)}{\sum_{g'}\exp\big(\cos(e_a,e_{s'_1})/\tau\big)}
  + \log\frac{\exp\big(\cos(e_{s_1},e_a)/\tau\big)}{\sum_{g'}\exp\big(\cos(e_{s_1},e_{a'})/\tau\big)}\right],
  \label{eq:xp}
\end{equation}
where $(a',s'_1)$ is the anchor and top sibling of group $g'$, and $g'$ ranges over the groups of the global batch except the other groups of the same document, which would be false negatives. The loss has the form of contrastive pre-training \citep{izacard2022contriever}; what \method{} changes is the positive. Contrastive methods pair a passage with a random crop, an adjacent span or a structural counterpart such as a title, whether or not the two texts are related. \method{} pairs a representative anchor with the passage of its document that is most related to it, so that the embedding must capture the relation between two different passages on one topic -- the relation between a query and a passage that answers it.

\subsection{Cross-Passage Content Supervision Signals}\label{sec:pred}

$\Lxp$ uses one sibling per anchor and constrains only embedding geometry. Two prediction terms supervise \emph{content}: an anchor's embedding should encode its related passages, and every embedding its own passage. Both use the decoder $h_\phi$, one freshly initialized Transformer layer sharing the encoder's token embeddings and MLM head. A linear map projects the conditioning embedding $e_x$ to $U(x)\in\mathbb{R}^{4\times d}$, inserted after [CLS] in the heavy view $m_H(y)$ of a target passage $y$, and $h_\phi$ predicts $q_t(\cdot\mid m_H(y),U(x))$ at each masked position $t\in M_H(y)$. One layer cannot recover 60\% of $y$ from the rest, so it must draw on $U(x)$.

$\Lcross$, the prediction term that \method{} introduces, asks the embedding of one passage to predict the masked tokens of another: those of each sibling from the anchor, and those of the anchor from each sibling. With the masked-language-modeling cross-entropy on the actual tokens $y_t$,
\begin{equation}
  C(y\mid x) = -\frac{1}{|M_H(y)|}\sum_{t\in M_H(y)}\log q_t\big(y_t\mid m_H(y),U(x)\big),
  \label{eq:rcross}
\end{equation}
and $\Lcross = \frac{1}{2n}\sum_{j=1}^{n}\big[C(s_j\mid a) + C(a\mid s_j)\big]$.
To predict a passage that it has never read, the anchor's embedding must encode what the two passages share; bottleneck auto-encoders, in contrast, reconstruct a passage from its own embedding (Section \S\ref{sec:related}).

\subsection{Within-Passage Content Supervision Signals}
Beyond its tokens, the encoder's reading of a passage, its output distribution at every masked position, describes what the passage is about. $\Lself$ distills this reading into the embedding: the encoder, which reads the light view, is the teacher, and the decoder, which reads the heavy view conditioned on the same passage's embedding, is the student: $\Lself = \tfrac12\big[D(a) + D(s_1)\big]$ with
\begin{equation}
  D(y) = \frac{1}{|M_H(y)|}\sum_{t\in M_H(y)}
  \mathrm{KL}\Big(\sg\big[p_t(\cdot\mid m_L(y))\big] \,\Big\|\, q_t\big(\cdot\mid m_H(y),U(y)\big)\Big),
  \label{eq:rself}
\end{equation}
where $p_t$ is the distribution of the encoder's MLM head and $\sg[\cdot]$ stops gradients through the targets. Every target comes from a view that saw more of $y$ than the student \citep{hinton2015distilling}. Both terms are token-weighted batch averages (Appendix~\ref{app:objective}).

An MLM loss $\Lmlm$ on the light views of $a$ and $s_1$ keeps the teacher's distributions grounded in the actual tokens, and the uniformity loss $\Lunif$ of \citet{wang2020uniformity} on the normalized anchor embeddings keeps them from collapsing. After pre-training, the projection, the decoder and the MLM head are discarded; the retriever is $f_\theta$ with mean pooling, so inference cost is unchanged. Per group, training costs five encoder passes and ten one-layer decoder passes.


\section{Experimental setup}
\label{sec:setup}


\paragraph{Pre-training corpora.} All models are pre-trained on English web documents from FineWeb \citep{penedo2024fineweb}. We take the first one million non-empty documents in the corpus's released order, tokenize them with the BERT uncased vocabulary shared by all our models, and split them into 128-token passages, keeping a trailing passage only if it has at least 64 tokens; documents with fewer than five passages cannot form a group and are discarded. From this pool we build three corpora. \textbf{100K} is a length-stratified sample of 100,000 documents, of which 40\%, 30\%, 21\% and 9\% have 5--8, 9--16, 17--32 and 33--64 passages. \textbf{1M} contains every document of the million with 5--64 passages, and \textbf{500K} is the subset of 1M that falls in a fixed random half of the million.

\paragraph{Continual pre-training from a strong retriever (main setting).} We initialize every model from E5-base-unsupervised \citep{wang2022e5}, a BERT-base encoder contrastively pre-trained on large-scale weakly supervised text pairs. E5 ships without an MLM head, so we initialize the head from BERT-base-uncased, which shares its vocabulary. Pre-training continues for   one pass over the groups of each corpus with batch size 512 (1 epoch over all the documents) and a peak learning rate of $10^{-6}$ with 5\% linear warm-up and linear decay; the encoder is frozen for the first 100 steps while the newly initialized components warm up (Table~\ref{tab:hparams}). Continual pre-training must not erase what the initial model already knows, so we apply a retention regularizer of the same form, weight and schedule to \emph{every} method, baselines included (Appendix~\ref{app:retention} lists the passages each method uses): a symmetric KL divergence between the in-batch similarity distributions of the current model and of the frozen initial model \citep{tung2019similarity,park2019relational}.
Its weight decays from 40 to 0 (Appendix~\ref{app:retention}) as $\lambda_{\text{ret}}(t) = 40\,\big(1-\sqrt{(t-1)/(T-1)}\big)$ over the $T$ training steps, which protects the initial model early and lets the pre-training objective dominate late. 

\paragraph{Pre-training from BERT.} To test whether the signal depends on a contrastive starting point, we also pre-train from BERT-base-uncased \citep{devlin2019bert}. \xPRT\ uses the same configuration as in the continual setting -- four siblings per anchor and every term of Eq.~\ref{eq:objective} -- but without the retention regularizer, since there is no pre-trained retriever to preserve. Training uses the 100K groups for six epochs (8,172 steps, batch 512) at learning rate $10^{-4}$, without a backbone freeze (Table~\ref{tab:hparams}). The groups are the same as in the continual setting and were selected with E5, so this regime tests the objective, not a selection-free pipeline. Because each run is cheap, we also use this regime for the ablations of \S\ref{sec:ablation}.

\paragraph{Baselines.} We compare with five retrieval-oriented pre-training objectives, using their official implementations adapted to our data pipeline. Four of them learn within a passage: RetroMAE \citep{xiao2022retromae}, masked auto-encoding with an enhanced one-layer decoder; DupMAE \citep{liu2023retromae2}, which adds a bag-of-words objective to RetroMAE; SimLM \citep{wang2023simlm}, replaced-language modeling through a representation bottleneck; and BoW-DPR \citep{ma2024bowdpr}, bag-of-words prediction without a decoder. Each of them trains on the anchor passage of every group -- the same passages that \xPRT\ uses as anchors, including the multiple anchors per document -- with its native masking rates and loss weights (Table~\ref{tab:baselines}). The fifth, Revela \citep{cai2026revela}, learns from other passages: a language model (Qwen2.5-0.5B with LoRA) predicts each passage while attending to the other passages of its group, weighted by the retriever's similarities. Each Revela group holds an anchor, its four siblings and eleven anchors of other documents. All baselines use the initialization, number of steps, batch size, learning-rate schedule and, in the continual setting, retention regularizer of \xPRT{}. We also report the initial model, E5, without further training.

\paragraph{Evaluation.} We evaluate the final checkpoint of every run, without fine-tuning, on the 41 tasks of MTEB(eng, v2) \citep{muennighoff2023mteb,enevoldsen2025mmteb}. The tasks span classification (8), clustering (8), pair classification (3), reranking (2), retrieval (10), semantic textual similarity (9) and summarization (1). Embeddings are mean-pooled and L2-normalized, inputs carry E5's \texttt{query:} and \texttt{passage:} prefixes, the maximum length is 512 tokens, and inference runs in bf16. We report the average main score over the 41 tasks (\emph{Avg.}) (Table~\ref{tab:retrieval-metrics}), the number of tasks on which a model is better or worse than the initial E5 (W/L), and nDCG@10, MRR@10, Recall@10 and Recall@100 averaged over the ten retrieval tasks. We select hyperparameters for the final config through validation loss that is disjoint from the evaluation benchmark which is contrastive perplexity on the MS MARCO passage-ranking dev set \citep{bajaj2016msmarco}.
    
    \begin{table}[th]
      \centering
      \caption{\textbf{Continual pre-training from E5.} Established objectives leave E5 essentially unchanged, whereas \method{} improves every task type at every corpus scale. Columns give the average main score per task type (number of tasks below; nDCG@10 for retrieval) and over all 41 tasks (Avg.); $\Delta$ is the change in Avg.\ with respect to E5, computed before rounding. Scores are $\times 100$, and the best score in each column and corpus is in bold. Win/loss counts and further retrieval metrics are in Table~\ref{tab:pertask} and Table~\ref{tab:retrieval-metrics}.}
      \label{tab:main}
      \small
      \setlength{\tabcolsep}{3.6pt}
      \begin{tabular}{llccccccccc}
        \toprule
        & & Class. & Clust. & Pair & Rerank & Retr. & STS & Summ. & Avg. & \\
        Corpus & Method & (8) & (8) & (3) & (2) & (10) & (9) & (1) & (41) & $\Delta$ \\
        \midrule
        -- & E5 (initial) & 55.72 & 46.06 & 79.19 & 44.80 & 43.26 & 63.36 & 31.46 & 53.06 & -- \\
        \midrule
        \multirow{6}{*}{100K} & RetroMAE & 55.96 & 45.70 & 79.25 & 44.70 & 42.41 & 63.69 & 32.19 & 52.92 & $-0.14$ \\
        & DupMAE & 55.95 & 45.90 & 79.24 & 44.69 & 42.44 & 63.73 & 32.13 & 52.97 & $-0.09$ \\
        & SimLM & 56.22 & 39.62 & \textbf{80.67} & 44.29 & 36.93 & 66.32 & 28.10 & 51.01 & $-2.05$ \\
        & BoW-DPR & 56.12 & 43.58 & 80.07 & 44.60 & 43.11 & 65.67 & 29.57 & 53.14 & $+0.08$ \\
        & Revela & 56.10 & 45.86 & 79.16 & \textbf{44.96} & 42.99 & 63.17 & 32.13 & 53.02 & $-0.05$ \\
        & \method{} & \textbf{56.74} & \textbf{46.75} & 79.84 & 44.83 & \textbf{45.78} & \textbf{66.85} & \textbf{32.87} & \textbf{54.86} & $\mathbf{+1.80}$ \\
        \midrule
        \multirow{6}{*}{500K} & RetroMAE & 55.99 & 45.81 & 79.31 & 44.69 & 42.38 & 63.90 & 32.05 & 52.99 & $-0.07$ \\
        & DupMAE & 56.00 & 45.93 & 79.32 & 44.66 & 42.32 & 63.90 & 32.16 & 53.00 & $-0.06$ \\
        & SimLM & \textbf{57.00} & 45.23 & \textbf{80.40} & 44.52 & 42.01 & 65.30 & 30.38 & 53.32 & $+0.26$ \\
        & BoW-DPR & 56.63 & 45.24 & 79.92 & 44.76 & 43.38 & 64.62 & 31.61 & 53.44 & $+0.38$ \\
        & Revela & 56.07 & 45.92 & 79.00 & \textbf{45.06} & 42.66 & 63.10 & 31.15 & 52.89 & $-0.17$ \\
        & \method{} & 56.67 & \textbf{47.04} & 79.91 & 44.87 & \textbf{45.67} & \textbf{67.04} & \textbf{32.93} & \textbf{54.93} & $\mathbf{+1.87}$ \\
        \midrule
        \multirow{6}{*}{1M} & RetroMAE & 56.13 & 45.88 & 79.39 & 44.66 & 41.94 & 64.07 & 32.40 & 52.97 & $-0.09$ \\
        & DupMAE & 56.11 & 45.73 & 79.39 & 44.65 & 41.99 & 64.03 & 32.41 & 52.94 & $-0.12$ \\
        & SimLM & 56.50 & 43.82 & 79.52 & 43.85 & 37.93 & 63.07 & 29.96 & 51.36 & $-1.70$ \\
        & BoW-DPR & \textbf{56.76} & 45.13 & 79.62 & 44.39 & 42.36 & 64.07 & 30.95 & 53.02 & $-0.04$ \\
        & Revela & 56.11 & 46.10 & 79.07 & \textbf{44.88} & 42.47 & 62.82 & 30.22 & 52.80 & $-0.26$ \\
        & \method{} & 56.68 & \textbf{47.15} & \textbf{79.89} & 44.87 & \textbf{45.33} & \textbf{67.13} & \textbf{33.00} & \textbf{54.89} & $\mathbf{+1.83}$ \\
        \bottomrule
      \end{tabular}
    \end{table}

\section{Results}
\label{sec:results}

We ask whether the relations between the passages of a document add to what a strong retriever and established objectives already capture (\S\ref{sec:main}), where the gains appear (\S\ref{sec:breadth}), and which ingredients produce them (\S\ref{sec:ablation} and \S\ref{sec:analysis}).

\subsection{Main results}\label{sec:main}
We conduct our experiments on E5, which was contrastively pre-trained on hundreds of millions of text pairs. We observe that well-established pre-training objectives could not improve it meaningfully (Table~\ref{tab:main}) in a continual pre-training setting. Across three corpora we observe that, the four within-passage objectives exhibit an average performance difference of about $-2.05$ to $+0.38$ points. Even the SimLM objective is unstable, losing 2.05 and 1.70 points at 100K and 1M data scales. Revela, which learns from other passages through a language model, ends 0.05 to 0.26 points below the E5 starting point. \method{}, with identical data and budget, improves E5 by 1.80, 1.87 and 1.83 points at 100K, 500K and 1M, on 34-36 of the 41 tasks (Table~\ref{tab:retrieval-metrics}) and on every task type at every scale. Furthermore, its nDCG@10 at 100K, 45.78, is 6.2\% above that of the strongest baseline, demonstrating the effectiveness of the proposed method.

One intuition that may help explain the result is that E5 as well as most of the pre-training objectives in literature lack direct supervision from the relation between different passages in a single document. E5's pairs are linked by web structure, such as a title and its body, and reconstructing a passage from its own embedding, in each of the four forms we tried, provides minimal complementary value. \method{} supervises the embedding with this intra-document relations directly. Our objective forces the anchor's embedding to single out its most related passage among passages of other documents, and learn sufficiently from the related passages' content for a one-layer decoder to predict them. Related passages alone do not suffice, as Revela shows. It trains on the same groups as \method{} -- an anchor, its four siblings and eleven anchors of other documents -- but the retriever's similarities only weight the attention of a language model that predicts each passage, so the retriever learns only indirectly, from how much each passage helps the language model. In \method{}, the relation acts on the embedding itself.

Crucially, the signal does not depend on the model used as the starting point. In a controlled ablation, we observe that pre-training BERT for six epochs on the 100K groups (Table~\ref{tab:bert}), \method{} scores 4.4 points above the strongest baseline, SimLM, with a 41\% higher retrieval nDCG@10 (29.21 against 20.68), and it is the best of the six models on 28 of the 41 tasks (Table~\ref{tab:pertask-bert}). Revela, the other objective that learns from other passages, ends 0.5 points below SimLM. The gains lie mostly in retrieval, STS, pair classification and clustering; classification and summarization are the only task types on which baselines remain ahead.

\subsection{Gains across task types and corpus scales}\label{sec:breadth}
The gains generalize across various types and corpora that \method{} is evaluated on. At 100K, gains are largest on STS (+3.5) and retrieval (+2.5 nDCG@10), which score whether two different texts mean the exact thing, and smallest on the reranking task (under 0.1), pair classification (+0.6) and clustering (+0.7). Within retrieval, the largest gains appear on tasks whose query and target are linked by an argument or a claim: Touch\'e-2020 (+6.6), ArguAna (+5.8), FEVER (+4.1), Climate-FEVER (+3.6) and HotpotQA (+3.0) (Table~\ref{tab:pertask}). The two regressions, TREC-COVID ($-1.4$) and FiQA ($-1.1$), are in specialized domains far from web text, and the two reranking tasks, which reorder short candidate lists by fine distinctions between similar questions (AskUbuntu) or by a user's clicks (MIND), are unchanged. The gain is also stable across scales: from 100K to 1M documents, with clustering and STS still rising and nDCG@10 slightly lower at 1M (45.33) than at 100K (45.78), as TREC-COVID and FiQA degrade (Table~\ref{tab:retrieval-scale}). We attribute the plateau to the continual budget, one pass at a learning rate of $10^{-6}$ under a decaying retention constraint.

\subsection{Curation of related passages is critical}\label{sec:ablation}

\begin{table}[t]
  \centering
  \caption{\textbf{Pre-training from BERT} (100K groups, six epochs, no retention regularizer). Columns as in Table~\ref{tab:main}; $\Delta$ is the change in Avg.\ relative to SimLM, the strongest baseline. \method{} and its variant without $\Rcross$ (Table~\ref{tab:ablation}) are means over three seeds. Per-task scores are in Table~\ref{tab:pertask-bert}.}
  \label{tab:bert}
  \small
  \setlength{\tabcolsep}{3.6pt}
  \begin{tabular}{lccccccccc}
    \toprule
    & Class. & Clust. & Pair & Rerank & Retr. & STS & Summ. & Avg. & \\
    Method & (8) & (8) & (3) & (2) & (10) & (9) & (1) & (41) & $\Delta$ \\
    \midrule
    RetroMAE & 51.43 & 36.96 & 72.28 & 38.60 & 19.65 & 54.37 & 33.28 & 41.96 & $-0.77$ \\
    DupMAE & 50.16 & 35.79 & 70.71 & 38.13 & 18.89 & 52.66 & \textbf{33.77} & 40.80 & $-1.93$ \\
    SimLM & 52.02 & 40.16 & 71.55 & 38.60 & 20.68 & 53.81 & 31.38 & 42.73 & -- \\
    BoW-DPR & \textbf{52.46} & 39.11 & 71.23 & 39.25 & 19.04 & 55.31 & 32.29 & 42.57 & $-0.16$ \\
    Revela & 51.25 & 40.28 & 69.88 & 39.43 & 17.45 & 56.14 & 31.35 & 42.24 & $-0.49$ \\
    \midrule
    \method{} & 51.33 & \textbf{44.29} & \textbf{76.20} & \textbf{40.77} & \textbf{29.21} & \textbf{59.83} & 26.23 & \textbf{47.12} & $\mathbf{+4.39}$ \\
    \quad without $\Rcross$ (\S\ref{sec:ablation}) & 50.96 & 44.17 & 75.85 & 40.74 & 28.46 & 59.27 & 25.25 & 46.67 & $+3.94$ \\
    \bottomrule
  \end{tabular}
\end{table}

\begin{table}[t]
\centering
\caption{\textbf{Ablations when pre-training from BERT} (100K groups, six epochs; $\times 100$): averages over 41 tasks and 10 retrieval tasks.}
\label{tab:ablation}
\small
\setlength{\tabcolsep}{3.1pt}
\begin{tabular}{llcrrr}
\toprule
& & \multicolumn{2}{c}{MTEB (41)} & \multicolumn{2}{c}{Retrieval (10)} \\
\cmidrule(lr){3-4}\cmidrule(lr){5-6}
& Variant & Avg. & $\Delta$ & nDCG@10 & $\Delta$ \\
\midrule
& \method{} (Eq.~\ref{eq:objective}, $\lambda_c=0.1$) & \textbf{47.12} & & 29.21 & \\
& without $\Rcross$ & 46.67 & $-0.45$ & 28.46 & $-0.75$ \\
\cmidrule(lr){2-6}
\multirow{4}{*}{\shortstack[l]{$\Lxp$}}
& random passage of the same document & 42.22 & $-4.90$ & 19.87 & $-9.34$ \\
& top sibling of another document & 40.80 & $-6.32$ & 18.03 & $-11.18$ \\
\cmidrule(lr){2-6}
& SimLM (strongest baseline) & 42.73 & $-4.39$ & 20.68 & $-8.53$ \\
\midrule
\cmidrule(lr){2-6}
\multirow{4}{*}{\emph{objective}}
    & without $\Lxp$ (prediction only) & 43.72 & $-3.4$ & 19.47 & $-9.74$ \\
& without $\Rself$ & 44.40 & $-2.72$ & 24.78 & $-4.4$ \\
& without $\Lmlm$ & 43.19 & $-3.94$ & 22.93 & $-6.29$ \\
& without $\Lunif$ & 46.73 & $-0.39$ & 28.45 & $-0.76$ \\
\midrule
\multirow{3}{*}{\shortstack[l]{\emph{no cross-}\\\emph{passage}\\\emph{signal}}}
& self-prediction, all five passages & 42.62 & $-4.50$ & 19.20 & $-10.01$ \\
& self-prediction, $a$ and $s_1$, $2\times$ steps & 42.75 & $-4.37$ & 18.80 & $-10.41$ \\
& RetroMAE, all five passages, $5\times$ steps & 41.73 & $-5.39$ & 19.21 & $-10.00$\\
\midrule
& \emph{siblings} ranked by E5 cosine only (no TF-IDF) & 44.82 & $-2.30$ & 25.88 & $-3.33$ \\
&  \emph{siblings} ranked by TF-IDF only (no E5 cosine) & 47.10 & -0.02& \textbf{29.89} & +0.67\\
\bottomrule
\end{tabular}
\end{table}

Table~\ref{tab:ablation} ablates \method{} in the BERT regime, isolating the contributions of each term in the objective and of sibling selection by removing ingredients one at a time. The decisive ingredient is the cross-passage signal from curated related passages. Removing cross-passage prediction alone lowers the average by 0.45 points and nDCG@10 by 0.75, on every task type (Table~\ref{tab:bert}). Removing the cross-passage signal altogether -- dropping $\Lcross$ and pairing each anchor with a random passage of its document instead of its selected sibling -- costs 4.9 points on average and 9.3 of nDCG@10, 32\% of the retrieval score, and falls below SimLM; a sibling from another document is worse still. More text cannot substitute for this signal: three controls that drop both cross-passage terms and give self-prediction or RetroMAE as much text as \method{} or more, with up to five times the training steps, still plateau at 18.8--19.2 nDCG@10.

Random pairs fail because web pages mix genuine content with widgets, link lists, and contact details (Figure~\ref{fig:motivation}; Table~\ref{tab:anchor-examples}) -- noise inherited by random same-document pairing, as used in contrastive pre-training \citep{izacard2022contriever,gao2022cocondenser}. Just as filtering noisy pairs benefits weakly supervised pre-training \citep{wang2022e5}, curating them is essential in the unsupervised setting.
Once the pairs are curated, the remaining design choices matter far less. Sibling ranking method changes the average by at most 2.3 points (TF-IDF alone even gives a slightly higher nDCG@10 than the fused ranking, $+0.67$), whereas a random passage costs 4.9. Among the loss terms, $\Lxp$ carries most of the gain ($-9.7$ nDCG@10 when removed), $\Lmlm$ is also necessary ($-3.9$ points), since it grounds the prediction targets in the actual tokens, and likewise $\Lself$ ($-2.7$ points) matters, removing $\Lunif$, by contrast, changes the average by only 0.4.

\subsection{What the cross-passage signal changes}\label{sec:analysis}
The cross-passage signal teaches the embedding which passages address one another, beyond which ones share words. ArguAna shows this most clearly: its target is the counter-argument of the query, it has one of the largest gains (+5.8 nDCG@10), and its debate pages, which collect points and counterpoints on one motion, resemble \method{}'s groups. Re-encoding its 1,406 queries in full precision (Appendix~\ref{app:analysis}), \method{} ranks a passage from the query's own debate first for 79.2\% of the queries, against 73.9\% for E5. The gain does not come from lexical matching: after splitting the queries into thirds by the IDF-weighted overlap between query and counter-argument, \method{} improves over E5 in every third (+3.7 to +8.8 nDCG@10), whereas BoW-DPR, whose objective is lexical, gains only in the highest-overlap third and loses 10.3 points in the lowest; CQADupstack-Gaming shows the same pattern. For the argument that legalizing drugs would reduce crime, E5 ranks first another argument of the same debate, which does not address crime, whereas \method{} ranks first the counter-argument that most drug crimes are committed by people on drugs (Table~\ref{tab:examples} gives six such queries).

\section{Related work}
\label{sec:related}

\textbf{Unsupervised pair construction:}
Lacking labelled pairs \citep{bajaj2016msmarco}, one family builds pseudo-relevant pairs from an unlabelled corpus and trains the pooled vector contrastively \citep{oord2018cpc}: inverse cloze pairs a sentence with its context \citep{lee2019ict,chang2020pretraining}, Contriever and SPIDER pair two crops of one document \citep{izacard2022contriever,ram2022spider}, coCondenser adds a corpus-aware term over a bottleneck \citep{gao2022cocondenser}, and weakly supervised mining filters noisy pairs at scale \citep{wang2022e5}; supervised pipelines then add hard negatives and cross-encoder filtering \citep{xiong2021ance,qu2021rocketqa}. In all of them the positive is fixed by position or by chance, so a pair counts as equally relevant whether its two passages share a subject or only a URL.

\textbf{Bottleneck auto-encoding:}
A second family reconstructs a passage through the pooled vector, weakening the decoder so that the vector is the only route to the target \citep{lu2021seed,gao2021condenser,xiao2022retromae,li2023challenging}. Later variants change what is reconstructed rather than where it comes from: a bag-of-words term \citep{liu2023retromae2}, replaced-language modeling \citep{wang2023simlm}, a lexicon bottleneck \citep{shen2023lexmae}, or the passage's own words with no decoder at all \citep{ma2024bowdpr}. Every target still comes from the passage just read.

\textbf{Supervision that leaves the passage:}
CoT-MAE reconstructs one neighbouring span alongside the input \citep{wu2023cotmae}, and Revela reaches the retriever through a language model that predicts each passage while attending to the others in its batch \citep{cai2026revela}. Our machinery is standard: soft targets from a better-informed view \citep{hinton2015distilling}, distillation that transfers geometry rather than logits \citep{park2019relational,tung2019similarity}, and InfoNCE read through alignment and uniformity \citep{wang2020uniformity}. What \xPRT\ adds is which passages to supervise (\S\ref{sec:selection}), which separates it from contrastive objectives whose positive is a random crop, an adjacent span or a structural counterpart, from auto-encoders whose targets never leave the encoded passage, and from Revela, where the relation between passages reaches the retriever through a language model's attention rather than through the embedding.

\section{Conclusion}
\label{sec:conclusion}
Established retrieval pre-training objectives leave unused how the passages of a document relate to one another. \method{} selects representative anchors and their most related passages, aligns their embeddings and predicts the related passages through the anchor's embedding. Under a controlled protocol, it is the only objective that improves E5 substantially, by 1.8--1.9 points on MTEB(eng, v2) and on every task type at every corpus scale, and from BERT it outperforms every baseline by 4.4 points or more; ablations trace the gain to the cross-passage signal from curated related passages.

\subsection*{AI use statement}
In this work, we used generative AI tools to propose and refine hypotheses, to give feedback on our experimental design, to help implement parts of the method and the evaluation harness, and to help interpret results. Additionally, we used generative AI tools to search for and summarize related work, to identify relevant literature and format references, to draft and edit parts of the manuscript for readability, and to suggest its structure.

We have reviewed all AI-assisted work. Every reference was checked against primary metadata from arXiv or the publisher rather than accepted from model output, because language models are known to fabricate citations, and references we could not verify were removed. Every number reported in this paper comes from a training or evaluation run that we executed and logged; no result was produced or estimated by a model. Where a model proposed an experiment or an explanation for a result, we ran the experiment and confirmed the outcome before any claim resting on it entered the paper, and AI-assisted code was read and tested by the authors. We take responsibility for the final content of this work, including text, claims or artifacts produced with the aid of generative AI.

\subsection*{Reproducibility statement}
Section~\ref{sec:method} states the objective in full, including every loss term and the procedure that selects anchors and their siblings, and Section~\ref{sec:setup} states how the three corpora are derived from FineWeb, which public checkpoints initialize each run, the optimization schedule and batch size, the retention regularizer applied identically to every method (Appendix~\ref{app:retention}), the configuration of each baseline (Appendix~\ref{app:baselines}), and the evaluation protocol. Table~\ref{tab:hparams} gives complete hyperparameter settings for the continual and the BERT regime, and Appendix~\ref{app:results} gives per-task scores for every run from which the averages in the main text are computed. Group selection is deterministic given a corpus and the frozen scorer (Appendix~\ref{app:selection}), so the training examples underlying every run can be regenerated exactly. Table~\ref{tab:bert} and the first rows of Table~\ref{tab:ablation} report means over three seeds. We provide the complete code base, covering group selection, all pre-training objectives and the evaluation harness, with README files for reproducing each table, as anonymous supplementary material.

\bibliography{iclr2027_conference}
\bibliographystyle{iclr2027_conference}

\appendix
%
%
%
%

\providecommand{\E}{\mathbb{E}}
\providecommand{\method}{xPRT}
\providecommand{\Lnce}{\mathcal{L}_{\mathrm{NCE}}}
\providecommand{\by}[2]{\overset{\mathrm{(#1)}}{#2}}
\providecommand{\proofstep}[1]{\par\smallskip\noindent\textit{#1}\enspace}
\makeatletter
\@ifundefined{reasons}{%
  \newenvironment{reasons}%
    {\par\smallskip\noindent\small
     \begin{tabular}{@{}l@{\hspace{0.6em}}p{\dimexpr\linewidth-3em\relax}@{}}}%
    {\end{tabular}\par\smallskip}}{}
\makeatother

\section{Theoretical analysis}
\label{app:theory}

Our experiments establish two facts. Cross-passage prediction outperforms
within-passage reconstruction (Table~\ref{tab:ablation}), and reconstruction-based
pretrainers such as RetroMAE and DupMAE end \emph{below} the retriever they are
initialized from (Table~\ref{tab:main}). The first fact is expected. The second is not:
a weaker training signal should yield a smaller gain, not a loss. The analysis explains
both through a single mechanism.
\begin{enumerate}
\item \textbf{The text pairs dictate what any of the objectives learns.} Under squared error or KL, the optimal predictor is the average
  target given the input. When the target is another passage of the same document, this
  average removes everything except what the two passages share: their subject, which is
  exactly what relevance depends on (Theorem~\ref{app:thm:cross}).
\item \textbf{Reconstruction makes the subject compete with passage-specific content.}
  When the decoder also reads a masked copy of the target and the conditioning vector has
  limited capacity, reconstruction pays for each kind of information in proportion to how
  often the copy hides it. Topical words recur within a passage, so the copy reveals the
  subject more often than passage-specific content, and the subject loses the competition
  for capacity. In the limit of an unrestricted encoder, the optimum carries no
  information about the subject at all (Theorem~\ref{app:thm:separation}). This is
  consistent with the degradation.
\item \textbf{The conditions are made explicit.} The results assume that sibling passages
  share nothing beyond their subject and that queries are encoded like passages. They
  therefore also indicate where cross-passage training should help less: siblings that
  share non-topical content, and queries unlike passages. This result highlights promising directions for future work by improving how cross-passage training pairs are constructed for self-supervised retrieval model training.
\end{enumerate}
The analysis does not imply that within-passage signals are useless: removing
$\mathcal{L}_{\mathrm{self}}$ from the full objective lowers performance
(Table~\ref{tab:ablation}). Cross-passage prediction supplies the subject; within-passage
prediction can complement it, but on its own it puts the subject at a disadvantage.
The results concern the optima of the objectives, not training dynamics, and
Theorems~\ref{app:thm:within} and~\ref{app:thm:separation} hold in a model of passages
rather than on web data. Section~\ref{app:results} states the setting and the results.
The proofs in Section~\ref{app:proofs} can be skipped without losing the argument.

\subsection{Setting and results}
\label{app:results}

\paragraph{Setting.}
Let $Y\in\{1,\dots,k\}$ be the \emph{subject} of a document, that is, what it is about,
with prior $\pi_y=P(Y=y)>0$. The \emph{subject posterior} of a text $x$ is the vector
$p(x)\in\mathbb{R}^k$ with entries $p_y(x)=P(Y=y\mid x)$. It records everything $x$
reveals about the subject, and nothing else.

Both objectives train a predictor $\psi$ that maps an input $a$ to features
$\varphi(t)\in\mathbb{R}^d$ of a target $t$. In \method{}, $\varphi(t)$ is the teacher's
token distribution at a masked position of $t$, and $\psi(a)$ is the decoder's output
given the conditioning vector $U(a)$. The two objectives differ only in the pair $(a,t)$:
\begin{itemize}
\item \emph{cross-passage:} $a$ is an anchor and $t$ is a sibling $s$ from the same
  document;
\item \emph{within-passage:} $a$ is the visible part of a passage and $t$ is its masked
  part.
\end{itemize}
For squared error, and for the KL divergence from teacher to student used in $D$, the
optimal predictor is the conditional mean (Lemma~\ref{app:lem:mean}):
\begin{equation*}
  \psi^*(a)\;=\;\E\big[\varphi(t)\mid a\big].
\end{equation*}
$\psi^*$ is what an objective asks for, independently of architecture and optimization.
Because one input can be followed by many targets, $\psi^*$ averages over them. What
survives this averaging depends on the target, and this is the mechanism behind every
result below.

The \emph{profile} of subject $y$ is its typical target, $\mu_y=\E[\varphi(t)\mid Y=y]$.
Collecting the profiles as the columns of $M=[\mu_1\,\cdots\,\mu_k]\in\mathbb{R}^{d\times k}$,
the vector $M\,p(a)=\sum_y p_y(a)\,\mu_y$ mixes the profiles in proportion to the
posterior of $a$.

\paragraph{Premises.}
\begin{description}
\item[(P1) Relevance is a shared subject.] A query and its relevant passage are
independent given their common subject, and an irrelevant passage is an independent
draw from the corpus. Every query and passage is a possible anchor, and the subject
posterior of a text does not depend on whether it is read as a query, a passage or an anchor.
\item[(P2) Clean pairs.] Selected pairs are independent given their subject:
  $a\perp s\mid Y$.
\end{description}
Under P1, the score
\begin{equation*}
  S^*(q,x)\;=\;\sum_{y}\frac{p_y(q)\,p_y(x)}{\pi_y}
\end{equation*}
orders candidates by their posterior probability of being relevant
(Lemma~\ref{app:lem:rank}), so no retriever can rank better.

P1 defines the retrieval task, as any analysis of retrieval must; its second sentence is
the transfer assumption that lets a pretraining optimum be applied to queries. P2 is the
property that anchor filtering and sibling ranking are responsible for. It is the
conditional-independence condition used to analyze self-supervised learning
\citep{tosh2021contrastive,lee2021predicting}. It does not require a passage to have a
single topic, since $Y$ may describe a mixture of topics, nor does it require a passage
to reveal its subject with certainty. It excludes sharing that is not about the subject,
such as a sentence split across adjacent passages or common boilerplate. Within-passage
pairs violate P2 by construction, and the contrast between the two objectives rests on
this.

\paragraph{Cross-passage prediction.}
The siblings that can follow an anchor differ in everything except their subject, so
averaging over them removes everything else.

\begin{theorem}[Cross-passage prediction retains exactly the subject]
\label{app:thm:cross}
Under P2, for every anchor $a$,
\begin{equation*}
  \psi^*(a)\;=\;M\,p(a).
\end{equation*}
Consequently:
\begin{enumerate}
\item[\textup{(i)}] \emph{Nothing else is retained:} $\psi^*(a)$ depends on $a$ only
  through $p(a)$.
\item[\textup{(ii)}] \emph{The subject is retained up to what targets distinguish:}
  for every query $q$ and passage $x$ (under P1),
  \[ \psi^*(q)^\top\psi^*(x)\;=\;\sum_{y,y'}p_y(q)\,p_{y'}(x)\,\langle\mu_y,\mu_{y'}\rangle, \]
  so the dot product scores subject overlap, weighted by the similarity of the subjects'
  profiles. If $\mu_1,\dots,\mu_k$ are linearly independent, a fixed reweighting $G$ gives
  $\psi^*(q)^\top G\,\psi^*(x)=S^*(q,x)$, the best possible ranking.
\end{enumerate}
\end{theorem}

Linear independence means that different subjects produce different targets on average. Theorem~\ref{app:thm:cross} is the cross-passage instance of the
observation of \citet{lee2021predicting} that predicting one view from another recovers
the latent the views share.

\begin{remark}[The decoder's view of the target]
\label{app:rem:copy}
In \method{} the decoder also reads a heavily masked view $\hat s$ of the sibling.
Part~(i) of Theorem~\ref{app:thm:cross} persists: under P2, the optimal predictor
$\E[\varphi(s)\mid a,\hat s]$ depends on $a$ only through $p(a)$.
\end{remark}

\paragraph{A model of passages.}
A comparison with reconstruction requires a passage to contain something besides its
subject. In the \emph{component model}, a passage is $x=(Y,\nu_x)$, where the
\emph{specific content} $\nu_x$ is independent of $Y$ and of the content of every other
passage. Each text (an anchor, a sibling, or the visible or masked part of a passage) is
generated from the content of its passage, with randomness independent of other
passages. Target features are additive in the two components:
\begin{equation}
  \E\big[\varphi(t)\mid a,\,Y,\,\nu_t\big]\;=\;c(Y)+d(\nu_t),
  \tag{CM}\label{app:eq:cm}
\end{equation}
where $\nu_t$ is the specific content of the passage that $t$ comes from. For cross
pairs, \eqref{app:eq:cm} only asks the conditional mean of $\varphi$ to be additive.
For within pairs it is an idealization: the features of the masked part depend on the
content of the passage, not on which tokens were masked. Cross pairs satisfy P2 in this
model.

\paragraph{Within-passage prediction.}
When the target is the masked part of the same passage, its specific content does not
average out, because the visible part predicts it.

\begin{theorem}[Within-passage prediction also retains passage-specific content]
\label{app:thm:within}
In the component model, for every pair $(a,t)$,
\begin{equation*}
  \psi^*(a)\;=\;M\,p(a)+\Delta(a),
  \qquad
  \Delta(a)\;=\;\E\big[d(\nu_t)\mid a\big]-\E\big[d(\nu_t)\big].
\end{equation*}
\begin{enumerate}
\item[\textup{(i)}] \emph{Cross-passage:} $\Delta\equiv0$, which recovers
  Theorem~\ref{app:thm:cross}.
\item[\textup{(ii)}] \emph{Within-passage:} for visible parts $v,v'$ with $p(v)=p(v')$,
  \[ \psi^*(v)-\psi^*(v')\;=\;\E\big[d(\nu)\mid v\big]-\E\big[d(\nu)\mid v'\big], \]
  where $\nu$ is the specific content of the passage each part comes from. Passages
  about the same subject receive different predictions whenever their visible parts
  indicate different specific content.
\end{enumerate}
\end{theorem}

\paragraph{Adding the bottleneck and the decoder's copy.}
Theorem~\ref{app:thm:within} omits two features of \method{}: the decoder also reads a
heavily masked copy of the target, and the conditioning vector has limited capacity.
With both, the vector is rewarded only for what the copy leaves out.

In the component model, let $Y$ and $\nu$ be uniform on $\mathbb{Z}_m$ with $m\ge2$, and
let a passage be its content, $x=(Y,\nu_x)$. The encoder maps its input to a random
vector $U$ with capacity $I(U;\text{input})\le B$, where $0<B\le\log m$. The decoder
reads $U$ and a masked copy $\hat t$ of the target $t$. The copy reveals $Y$ with
probability $\beta_Y$ and $\nu_t$ with probability $\beta_\nu$, independently of each
other and of everything else, with $0<\beta_Y,\beta_\nu<1$. The decoder predicts $t$
under the log loss. By Lemma~\ref{app:lem:reward}, the loss of the best decoder is
$H(t\mid\hat t)-R$, where the \emph{reward}
\begin{equation*}
  R\;=\;I\big(U;\,t\mid\hat t\big)
\end{equation*}
is the information the vector adds about the target beyond the copy. In the
cross-passage objective the input is an anchor $a$ and the target is a sibling $s$. In
the within-passage objective the input and the target are both the passage $x$.

\begin{theorem}[Under a bottleneck, reconstruction disfavors the subject]
\label{app:thm:separation}
In this model:
\begin{enumerate}
\item[\textup{(i)}] \emph{Cross-passage.} For every encoder,
  $R_{\mathrm{cross}}=(1-\beta_Y)\,I(U;Y)$. Every maximizer has $I(U;Y)=B$: it spends its
  whole capacity on the subject.
\item[\textup{(ii)}] \emph{Within-passage: competition.} Suppose $U$ describes the subject
  and the specific content separately: $U=(U_Y,U_\nu)$, where $U_Y$ is a randomized
  function of $Y$ and $U_\nu$ of $\nu_x$, each with its own randomness. Let
  $b_Y=I(U;Y)$ and $b_\nu=I(U;\nu_x)$, with $b_Y+b_\nu\le B$. Then
  \[ R_{\mathrm{within}}\;=\;(1-\beta_Y)\,b_Y+(1-\beta_\nu)\,b_\nu, \]
  so each component is paid in proportion to how often the copy hides it. If
  $\beta_Y>\beta_\nu$, every maximizer has $b_Y=0$ and $b_\nu=B$.
\item[\textup{(iii)}] \emph{Within-passage: unrestricted encoders.} Every maximizer of
  $R_{\mathrm{within}}$ has $I(U;x)=B$ and $I(U;Y)=I(U;\nu_x)=0$, for every
  $\beta_Y,\beta_\nu\in(0,1)$.
\item[\textup{(iv)}] \emph{Retrieval.} Let a query be a passage $q=(Y,\nu_q)$ with fresh
  specific content, ranked against $N$ candidates: one relevant passage with the same
  subject and $N-1$ independent draws from the corpus. Let the retrieval embedding be any
  function of $U$, possibly randomized, and let the score be arbitrary, with ties broken
  at random. For every maximizer in (iii), the rank of the relevant passage is uniform on
  $\{1,\dots,N\}$. For every cross-passage maximizer with $B=\log m$, the subject is a
  function of $U$, and ranking by agreement of subjects is the best possible ranking.
\end{enumerate}
All maxima are attained.
\end{theorem}


\paragraph{Reading Theorem~\ref{app:thm:separation}.}
Part~(ii) is the general picture. Reconstruction pays for each kind of information in
proportion to how often the decoder's copy hides it. Topical words recur within a
passage, so the copy reveals the subject more often than a specific fact
($\beta_Y>\beta_\nu$), and the subject loses the competition for capacity. Part~(iii) is
the limit: an encoder that may also combine the two components drops the subject
entirely, using codes such as $(Y+\nu_x)\bmod m$ that are informative only together with
the copy. Cross-passage prediction involves no such competition. The sibling's specific
content cannot be predicted from the anchor, so the subject is the only information worth
encoding. The comparison with Theorem~\ref{app:thm:within} isolates the cause: without
the copy, the within-passage optimum contains the subject term $M\,p(v)$; with it, the
subject is at a disadvantage.

\paragraph{Scope.}
The results describe the optima of the objectives, not training dynamics.
Theorems~\ref{app:thm:within} and~\ref{app:thm:separation} hold in the component model,
and Theorem~\ref{app:thm:separation} abstracts the decoder as predicting the whole target
rather than a token distribution. The analysis concerns the predictive pathway, meaning
the conditioning vector and the decoder. Retrieval in \method{} uses the mean-pooled
$e$, which $\Lxp$ shapes directly. P1 treats passages that share a subject as equally relevant, so the analysis does not
model the value of passage-specific content for ranking among them. We attribute the
contribution of $\mathcal{L}_{\mathrm{self}}$ to this, but the theory does not establish
it.

\subsection{Proofs}
\label{app:proofs}

Throughout, masks, dropout and the teacher use fresh randomness for every text. Texts
take countably many values, $\varphi$ is bounded, and in
Theorem~\ref{app:thm:separation} all variables take finitely many values.

\subsubsection*{The optimal predictor and the optimal ranking}

\begin{lemma}[The optimal predictor is the conditional mean]
\label{app:lem:mean}
Let $\psi^*(a)=\E[\varphi(t)\mid a]$. For every predictor $\psi$:
\begin{enumerate}
\item[\textup{(i)}]
  $\E\|\varphi(t)-\psi(a)\|^2=\E\|\varphi(t)-\psi^*(a)\|^2+\E\|\psi^*(a)-\psi(a)\|^2$;
\item[\textup{(ii)}] if $\varphi(t)$ and $\psi(a)$ are distributions over the vocabulary,
  $\E\,\mathrm{KL}\big(\varphi(t)\,\|\,\psi(a)\big)
   =\E\,\mathrm{KL}\big(\varphi(t)\,\|\,\psi^*(a)\big)
   +\E\,\mathrm{KL}\big(\psi^*(a)\,\|\,\psi(a)\big)$.
\end{enumerate}
In both cases the first term does not depend on $\psi$, and the second vanishes only at
$\psi=\psi^*$.
\end{lemma}

\begin{proof}
\proofstep{Part (i).}
\begin{align*}
  \E\|\varphi(t)-\psi(a)\|^2
  &\by{a}{=}\;\E\|\varphi(t)-\psi^*(a)\|^2+\E\|\psi^*(a)-\psi(a)\|^2
     +2\,\E\big[(\varphi(t)-\psi^*(a))^\top(\psi^*(a)-\psi(a))\big]\\
  &\by{b}{=}\;\E\|\varphi(t)-\psi^*(a)\|^2+\E\|\psi^*(a)-\psi(a)\|^2.
\end{align*}
\begin{reasons}
(a) & Add and subtract $\psi^*(a)$, then expand the square. \\
(b) & Tower property: the second factor is a function of $a$, and
      $\E[\varphi(t)-\psi^*(a)\mid a]=0$.
\end{reasons}

\proofstep{Part (ii).}
\begin{align*}
  \E\,\mathrm{KL}\big(\varphi(t)\,\|\,\psi(a)\big)
  &\by{a}{=}\;\E\sum_v\varphi_v(t)\log\frac{\varphi_v(t)}{\psi^*_v(a)}
     +\E\sum_v\varphi_v(t)\log\frac{\psi^*_v(a)}{\psi_v(a)}\\
  &\by{b}{=}\;\E\,\mathrm{KL}\big(\varphi(t)\,\|\,\psi^*(a)\big)
     +\E\sum_v\psi^*_v(a)\log\frac{\psi^*_v(a)}{\psi_v(a)}.
\end{align*}
\begin{reasons}
(a) & Add and subtract $\log\psi^*_v(a)$. \\
(b) & Tower property: the log-ratio is a function of $a$, and
      $\E[\varphi_v(t)\mid a]=\psi^*_v(a)$.
\end{reasons}
The last sum is $\mathrm{KL}(\psi^*(a)\,\|\,\psi(a))$.
\end{proof}

\begin{lemma}[The optimal ranking]
\label{app:lem:rank}
Assume P1. Fix a query $q$ and $N$ candidates $x_1,\dots,x_N$: one relevant passage at
a uniformly random index $J$, and $N-1$ independent draws from the corpus. Then
\begin{equation*}
  P\big(J=j\mid q,x_1,\dots,x_N\big)\;=\;\frac{S^*(q,x_j)}{\sum_{i}S^*(q,x_i)}.
\end{equation*}
Ranking by $S^*$ therefore orders candidates by their posterior probability of being
relevant. In particular, placing the maximizer of $S^*$ first maximizes the probability
that the relevant passage is ranked first.
\end{lemma}

\begin{proof}
Given $J=j$, the candidates are independent, with $x_j\sim P(\cdot\mid q,\mathrm{rel})$
and $x_i\sim P(\cdot)$ for $i\neq j$. Since $J$ is uniform,
\begin{equation*}
  P\big(J=j\mid q,x_1,\dots,x_N\big)\;\propto\;
  P(x_j\mid q,\mathrm{rel})\prod_{i\neq j}P(x_i)
  \;=\;\Big(\prod_i P(x_i)\Big)\,\frac{P(x_j\mid q,\mathrm{rel})}{P(x_j)}.
\end{equation*}
The prefactor does not depend on $j$. For the ratio,
\begin{align*}
  \frac{P(x\mid q,\mathrm{rel})}{P(x)}
  &\by{a}{=}\;\sum_y p_y(q)\,\frac{P(x\mid Y=y)}{P(x)}
  \;\by{b}{=}\;\sum_y\frac{p_y(q)\,p_y(x)}{\pi_y}\;=\;S^*(q,x).
\end{align*}
\begin{reasons}
(a) & P1: given its subject, the relevant passage is independent of the query, so
      $P(x\mid q,\mathrm{rel})=\sum_y P(Y=y\mid q)\,P(x\mid Y=y)$. \\
(b) & Bayes' rule: $P(x\mid Y=y)/P(x)=p_y(x)/\pi_y$.
\end{reasons}
Normalizing over $j$ gives the formula. A ranking that places index $j$ first is
correct with probability $P(J=j\mid q,x_1,\dots,x_N)$, which is largest at the maximizer
of $S^*$.
\end{proof}

\subsubsection*{Cross-passage and within-passage prediction}

\begin{proof}[Proof of Theorem~\ref{app:thm:cross}]
\begin{align*}
  \psi^*(a)\;=\;\E\big[\varphi(s)\mid a\big]
  &\by{a}{=}\;\sum_y P(Y=y\mid a)\;\E\big[\varphi(s)\mid a,\,Y=y\big]\\
  &\by{b}{=}\;\sum_y p_y(a)\;\E\big[\varphi(s)\mid Y=y\big]\\
  &\by{c}{=}\;\sum_y p_y(a)\,\mu_y\;=\;M\,p(a).
\end{align*}
\begin{reasons}
(a) & Law of total expectation over $Y$. \\
(b) & P2: given $Y$, the sibling is independent of the anchor. \\
(c) & Definitions of $\mu_y$ and $M$.
\end{reasons}
Part~(i) is immediate. 

\begin{align*}
     \psi^*(q)^\top\psi^*(x)\;=\;\sum_{y,y'}p_y(q)\,p_{y'}(x)\,\langle\mu_y,\mu_{y'}\rangle
\end{align*}

For part~(ii), linear independence of the columns of $M$ makes
$M^\top M$ invertible, so $M^+=(M^\top M)^{-1}M^\top$ satisfies $M^+M=I_k$ and
$p(x)=M^+\psi^*(x)$. By P1, this holds for queries and passages as well as anchors. Let
$\Pi=\mathrm{diag}(\pi)$ and $G=(M^+)^\top\Pi^{-1}M^+$. Then
\begin{equation*}
  \psi^*(q)^\top G\,\psi^*(x)
  \;=\;\big(M^+\psi^*(q)\big)^\top\Pi^{-1}\big(M^+\psi^*(x)\big)
  \;=\;p(q)^\top\Pi^{-1}p(x)\;=\;S^*(q,x),
\end{equation*}
and Lemma~\ref{app:lem:rank} completes the proof.
\end{proof}

\begin{proof}[Proof of Remark~\ref{app:rem:copy}]
\begin{align*}
  \E\big[\varphi(s)\mid a,\hat s\big]
  &\by{a}{=}\;\sum_y P(Y=y\mid a,\hat s)\;\E\big[\varphi(s)\mid a,\hat s,\,Y=y\big]\\
  &\by{b}{=}\;\sum_y P(Y=y\mid a,\hat s)\;\E\big[\varphi(s)\mid\hat s,\,Y=y\big],
  \\[2pt]
  P(Y=y\mid a,\hat s)
  &\by{c}{=}\;\frac{p_y(a)\,P(\hat s\mid Y=y)}{\sum_{y'}p_{y'}(a)\,P(\hat s\mid Y=y')}.
\end{align*}
\begin{reasons}
(a) & Law of total expectation over $Y$. \\
(b) & $\hat s$ is generated from $s$ with fresh randomness, so P2 gives
      $(s,\hat s)\perp a\mid Y$. \\
(c) & Bayes' rule, using $\hat s\perp a\mid Y$.
\end{reasons}
Both factors depend on $a$ only through $p(a)$.
\end{proof}

\begin{proof}[Proof of Theorem~\ref{app:thm:within}]
\begin{align*}
  \psi^*(a)\;=\;\E\big[\varphi(t)\mid a\big]
  &\by{a}{=}\;\E\Big[\,\E\big[\varphi(t)\mid a,Y,\nu_t\big]\;\Big|\;a\Big]\\
  &\by{b}{=}\;\E\big[c(Y)+d(\nu_t)\mid a\big]\\
  &\by{c}{=}\;\sum_y p_y(a)\,c(y)+\E\big[d(\nu_t)\mid a\big]\\
  &\by{d}{=}\;\sum_y p_y(a)\,\mu_y-\E\big[d(\nu_t)\big]+\E\big[d(\nu_t)\mid a\big]
  \;=\;M\,p(a)+\Delta(a).
\end{align*}
\begin{reasons}
(a) & Tower property. \\
(b) & \eqref{app:eq:cm}. \\
(c) & Linearity, then total expectation over $Y$. \\
(d) & $\mu_y=c(y)+\E[d(\nu_t)\mid Y=y]=c(y)+\E[d(\nu_t)]$, because $\nu_t\perp Y$;
      also $\sum_y p_y(a)=1$.
\end{reasons}

\proofstep{Part (i).} The anchor is generated from $(Y,\nu_a)$ and randomness of its own
passage, all of which are independent of $\nu_s$. Hence $\nu_s\perp a$, so
$\E[d(\nu_s)\mid a]=\E[d(\nu_s)]$ and $\Delta\equiv0$. The same independence gives
$a\perp s\mid Y$, so P2 holds in the model.

\proofstep{Part (ii).} If $p(v)=p(v')$, then $M\,p(v)=M\,p(v')$, and subtracting the two
expressions for $\psi^*$ leaves the difference of the $\Delta$ terms.
\end{proof}

\subsubsection*{The bottleneck and the decoder's copy}

\begin{lemma}[Information in a masked view]
\label{app:lem:views}
Let $w=(w_1,w_2)$, and let $\hat w$ reveal $w_1$ with probability $\beta_1$ and $w_2$
with probability $\beta_2$, through indicators $E=(E_1,E_2)$ that are independent of each
other and of $(U,w)$. Then
\begin{align}
  I(U;\hat w)&=\beta_1\beta_2\,I(U;w)+\beta_1(1-\beta_2)\,I(U;w_1)+(1-\beta_1)\beta_2\,I(U;w_2),
  \label{app:eq:views}\\
  I(U;w\mid\hat w)&=I(U;w)-I(U;\hat w).
  \label{app:eq:views2}
\end{align}
\end{lemma}

\begin{proof}
\begin{align*}
  I(U;\hat w)
  \;\by{a}{=}\;I(U;\hat w,E)
  \;\by{b}{=}\;I(U;E)+I(U;\hat w\mid E)
  \;\by{c}{=}\;\sum_e P(E=e)\,I(U;\hat w\mid E=e).
\end{align*}
\begin{reasons}
(a) & $E$ can be read off $\hat w$, since the mask symbol differs from every value. \\
(b) & Chain rule. \\
(c) & $I(U;E)=0$, because $E\perp U$.
\end{reasons}
Conditioning on $E=e$ leaves the law of $(U,w)$ unchanged. The four values of $e$ reveal
$w$, $w_1$, $w_2$ and nothing, respectively, which gives \eqref{app:eq:views}. For
\eqref{app:eq:views2}, the chain rule gives
$I(U;w)+I(U;\hat w\mid w)=I(U;w,\hat w)=I(U;\hat w)+I(U;w\mid\hat w)$. Moreover
$I(U;\hat w\mid w)=0$, because $\hat w$ is a function of $(w,E)$ and $E\perp(U,w)$.
\end{proof}

\begin{lemma}[The loss falls by exactly the reward]
\label{app:lem:reward}
For a decoder that reads $(\hat t,U)$ and outputs $q(\cdot\mid\hat t,U)$,
\begin{equation*}
  \E\big[-\log q(t\mid\hat t,U)\big]
  \;=\;H(t\mid\hat t)-I(U;t\mid\hat t)
  +\E\,\mathrm{KL}\big(P(\cdot\mid\hat t,U)\,\|\,q(\cdot\mid\hat t,U)\big).
\end{equation*}
$H(t\mid\hat t)$ does not depend on the encoder, and the KL term vanishes for the best
decoder.
\end{lemma}

\begin{proof}
The expected log loss equals $H(t\mid\hat t,U)$ plus the expected KL term, and
$H(t\mid\hat t,U)=H(t\mid\hat t)-I(U;t\mid\hat t)$ by the definition of conditional
mutual information.
\end{proof}

\begin{proof}[Proof of Theorem~\ref{app:thm:separation}]
Let $\mathrm{Er}(S)$ denote $S$ passed through an erasure channel with erasure
probability $1-B/\log m$. For $S$ uniform on $\mathbb{Z}_m$,
$I(\mathrm{Er}(S);S)=(B/\log m)\,H(S)=B$.

\proofstep{Part (i).} Let $U=U(a)$.
\begin{align*}
  R_{\mathrm{cross}}\;=\;I(U;s\mid\hat s)
  &\by{a}{=}\;I(U;s)-I(U;\hat s)\\
  &\by{b}{=}\;I(U;s)-\beta_Y\beta_\nu\,I(U;s)-\beta_Y(1-\beta_\nu)\,I(U;Y)
     -(1-\beta_Y)\beta_\nu\,I(U;\nu_s)\\
  &\by{c}{=}\;(1-\beta_Y)\,I(U;Y).
\end{align*}
\begin{reasons}
(a) & Equation~\eqref{app:eq:views2}. \\
(b) & Equation~\eqref{app:eq:views} with $w=(Y,\nu_s)$. \\
(c) & $\nu_s\perp(U,Y)$, so $I(U;s)=I(U;Y)+I(U;\nu_s\mid Y)=I(U;Y)$ and
      $I(U;\nu_s)=0$.
\end{reasons}
$Y$ is a function of $a$, so $I(U;Y)\le I(U;a)\le B$. Hence
$R_{\mathrm{cross}}\le(1-\beta_Y)B$, with equality if and only if $I(U;Y)=B$. The
encoder $U=\mathrm{Er}(Y)$ attains the bound, with $I(U;a)=I(U;Y)=B$. So the maximum is
$(1-\beta_Y)B$, and every maximizer has $I(U;Y)=B$.

\proofstep{Parts (ii) and (iii).} Let $U=U(x)$.
\begin{align*}
  R_{\mathrm{within}}\;=\;I(U;x\mid\hat x)
  &\by{a}{=}\;(1-\beta_Y\beta_\nu)\,I(U;x)-\beta_Y(1-\beta_\nu)\,I(U;Y)
     -(1-\beta_Y)\beta_\nu\,I(U;\nu_x).
\end{align*}
\begin{reasons}
(a) & Equations~\eqref{app:eq:views2} and~\eqref{app:eq:views} with $w=(Y,\nu_x)$.
\end{reasons}

\emph{Separable codes (part ii).} The pairs $(U_Y,Y)$ and $(U_\nu,\nu_x)$ are independent,
so $I(U;x)=b_Y+b_\nu$, $I(U;Y)=b_Y$ and $I(U;\nu_x)=b_\nu$. Substituting,
\begin{align*}
  R_{\mathrm{within}}
  &\by{b}{=}\;\big[1-\beta_Y\beta_\nu-\beta_Y(1-\beta_\nu)\big]\,b_Y
     +\big[1-\beta_Y\beta_\nu-(1-\beta_Y)\beta_\nu\big]\,b_\nu\\
  &\by{c}{=}\;(1-\beta_Y)\,b_Y+(1-\beta_\nu)\,b_\nu\\
  &\by{d}{=}\;(1-\beta_\nu)(b_Y+b_\nu)-(\beta_Y-\beta_\nu)\,b_Y
  \;\by{e}{\le}\;(1-\beta_\nu)\,B.
\end{align*}
\begin{reasons}
(b) & Collect the terms in $b_Y$ and $b_\nu$. \\
(c) & Simplify each bracket. \\
(d) & Rearrange. \\
(e) & $b_Y+b_\nu\le B$, and $(\beta_Y-\beta_\nu)\,b_Y\ge0$ when $\beta_Y>\beta_\nu$.
\end{reasons}
Equality in (e) holds if and only if $b_Y=0$ and $b_\nu=B$. The code with constant $U_Y$
and $U_\nu=\mathrm{Er}(\nu_x)$ attains it.

\emph{Unrestricted encoders (part iii).} In (a), all three coefficients are positive,
$I(U;x)\le B$, and mutual information is non-negative. Hence
$R_{\mathrm{within}}\le(1-\beta_Y\beta_\nu)B$, with equality if and only if $I(U;x)=B$
and $I(U;Y)=I(U;\nu_x)=0$. Let $S=(Y+\nu_x)\bmod m$ and $U=\mathrm{Er}(S)$. For every
value of $Y$, $S$ is uniform, so $S\perp Y$; likewise $S\perp\nu_x$. Since $U$ depends
on $x$ only through $S$, $I(U;x)=I(U;S)=B$, while $I(U;Y)\le I(S;Y)=0$ and
$I(U;\nu_x)\le I(S;\nu_x)=0$. The bound is therefore attained, and every maximizer
satisfies the equality conditions.

\proofstep{Part (iv).} Write $z=g(U)$ for the retrieval embedding, with fresh
randomness for each text, and $x^+$ for the relevant passage.

\emph{Maximizer in part (iii).} Since $x^+$ is distributed as a corpus passage, part~(iii)
gives $I(z(x^+);Y)\le I(U(x^+);Y)=0$, so $z(x^+)\perp Y$. The query's specific content
$\nu_q$ is independent of $(z(x^+),Y)$, hence $z(x^+)\perp(Y,\nu_q)=q$. The negatives are
independent of $(q,x^+)$ and have the same law as $x^+$. Given $q$, the $N$ candidate
embeddings are therefore i.i.d., their scores are exchangeable, and with random
tie-breaking the rank of $x^+$ is uniform on $\{1,\dots,N\}$.

\emph{Cross-passage maximizer with $B=\log m$.} Part~(i) gives $I(U;Y)=\log m=H(Y)$, so
$Y=f(U)$ almost surely for some function $f$. Every passage reveals its subject, so
$p(x)$ is the point mass at the subject of $x$ and
$S^*(q,x)=m\cdot\mathbf{1}[Y_q=Y_x]$. Ranking by $\mathbf{1}[f(U(q))=f(U(x))]$ is thus
ranking by $S^*$, which is optimal by Lemma~\ref{app:lem:rank}. P1 holds in the model,
because the query and the relevant passage are independent given $Y$.
\end{proof}

\section{Implementation Details}\label{app:impl}

\begin{table}[h]
  \centering
  \caption{\textbf{Hyperparameters of \method{}} in the two regimes.}
  \label{tab:hparams}
  \small
  \begin{tabular}{p{3.8cm}p{4.4cm}p{4.4cm}}
    \toprule
    & Continual from E5 (default) & From BERT \\
    \midrule
    Initialization & E5-base-unsupervised; MLM head from BERT-base-uncased & BERT-base-uncased \\
    Corpora & 100K, 500K, 1M (one pass) & 100K (six epochs) \\
    Passage length / min.\ tail & 128 / 64 tokens & 128 / 64 tokens \\
    Siblings per anchor used & 4 & 4 \\
    Decoder & fresh 1-layer Transformer $h_\phi$ & same \\
    Conditioning vectors; source & 4; linear map of the mean-pooled embedding & same \\
    Masks (light / heavy) & 15\% / 60\%, nested & same \\
    Prediction losses & cross: masked-token cross-entropy; self: full-vocabulary KL, $\tau=1$ & same \\
    Terms (weights) & $\Lxp$ 1, $\Rcross$ 0.1, $\Rself$ 1, $\Lmlm$ 1, $\Lunif$ 0.01 & same \\
    $\Lxp$ negatives & 511 in global batch, same-document masked & same \\
    Retention regularizer & Eq.~\ref{eq:retention}, weight $40\to0$, $\tau_r=0.075$ & none \\
    Optimizer & AdamW (wd 0.01, $\epsilon=10^{-4}$); Adam for biases/LayerNorm & same \\
    Peak LR; warm-up; decay & $10^{-6}$; 5\% of steps; linear & $10^{-4}$; 68 steps; linear \\
    Batch & 512 & 512 \\
    Backbone freeze & first 100 steps & none \\
    Number of epochs & 1 & 6 \\
    Precision; clipping; seed & bf16; 1.0; 42, 43, 44 & bf16; 1.0; 42, 43, 44 \\
    Hardware & 32 $\times$ AMD MI300X & 8 $\times$ AMD MI300X \\
    \bottomrule
  \end{tabular}
\end{table}

Inputs are \texttt{[CLS] query: $a$ [SEP]} for anchors and \texttt{[CLS] passage: $s$ [SEP]} for siblings. The special tokens and the prefix tokens are never masked. Masked tokens are always replaced by [MASK] (no random or unchanged replacement). One pass visits every group exactly once in a random order that is shared across GPUs. The ablations of Table~\ref{tab:ablation} keep the global batch of 512, without gradient accumulation, on 8, 16 or 32 GPUs.

\section{Group Selection}\label{app:selection}

\paragraph{Scores.} Let $\hat e_i$ be the normalized E5 embedding of passage $i$, computed on \texttt{passage: } followed by the passage, with mean pooling.
\begin{itemize}
  \item \textbf{Centrality.} Passage $i$ of document $d$ has centrality $s_i = \hat e_i^{\top}\bar e_d/\lVert\bar e_d\rVert$, where $\bar e_d$ is the mean of $\hat e_j$ over the passages of $d$, including $i$.
  \item \textbf{Informativeness.} $\iota_i$ is the mean of $\mathrm{idf}(t)=\ln\frac{N+1}{\mathrm{df}(t)+1}+1$ over the unique tokens of $i$. Here $N = 5{,}150{,}388$ is the number of passages in the first million documents and $\mathrm{df}(t)$ counts the passages that contain token $t$.
  \item \textbf{Rank fusion.} Ranks are computed within a document, in descending order, and tied scores share the average rank. The fused score of a candidate is
  \[
    \mathrm{RRF} = \frac{2}{10 + r_{\text{dense}}} + \frac{1}{10 + r_{\text{lex}}}.
  \]
\end{itemize}

\paragraph{Selection.}
\begin{itemize}
  \item \textbf{Anchors} are the passages with $s_i > 0.825$. Up to ten are kept per document, ranked by the fusion of $s_i$ and $\iota_i$; ties are broken by $s_i$, then $\iota_i$, then position.
  \item \textbf{Siblings.} For an anchor $a$, each other passage $c$ is ranked by the fusion of $\hat e_a^{\top}\hat e_c$ and the cosine between the TF-IDF vectors (raw counts $\times$ idf) of $a$ and $c$. The top four are kept.
\end{itemize}

\paragraph{Example.} The group in Figure~\ref{fig:method} comes from a 10-passage blog post about Chromebooks.
\begin{itemize}
  \item The anchor is passage 1, on account synchronization and sandboxing; its centrality is 0.842.
  \item The siblings, in rank order, are passages 2 (operating-system updates and security, continuing the anchor; E5 cosine 0.777), 6 (losing connectivity on a train; 0.750), 7 (cloud-centric machines for businesses; 0.750) and 4 (an audio problem; 0.646).
\end{itemize}

\paragraph{What the filters remove, and what they miss.} The passage least related to a representative anchor is often page furniture. On the page of Figure~\ref{fig:motivation} it is a dictionary widget. On an article about the Ford F-Series it is a list of music-video links, on a policy report on carbon capture a list of press contacts, and on an essay about digital magazines a stock ticker. These passages have an E5 cosine of 0.38--0.45 to the anchor, whereas the most related passage of the same page has 0.85--0.95. They also have low centrality (0.46--0.51) and never become anchors. The anchor filter, however, measures how representative a passage is of its page, not its quality. On a site-policy page, the passages of the legal disclaimer are central (0.95) and pass the filter; the passage least related to them is a list of recent visitors (cosine 0.48). Conversely, the introduction of an interview with a band fails the filter (0.79), presumably because the questions and answers that fill the rest of the page dominate its centroid. Table~\ref{tab:anchor-examples} shows these passages.

\begin{table}[h]
  \centering
  \caption{\textbf{Passages that pass and fail the anchor filter} on seven FineWeb pages ($s$: centrality; $\tau_a = 0.825$; number of passages in parentheses). On the first five pages, the failing passage is also the one least related to the passing passage (E5 cosine 0.38--0.45). The last two pages show what the filter misses: a legal disclaimer passes, and an informative introduction fails. Passages are abridged (ellipses) and shown as the lowercased model input; contact names and phone numbers are elided.}
  \label{tab:anchor-examples}
  \footnotesize
  \begin{tabular}{@{}>{\raggedright\arraybackslash}p{1.7cm}>{\raggedright\arraybackslash}p{5.6cm}>{\raggedright\arraybackslash}p{5.6cm}@{}}
    \toprule
    Page & Passes the filter & Fails the filter \\
    \midrule
    Baseball no-hitter article (47) & ($s=0.94$) \ldots{} a game in which a team allows no runners to reach base at all is a perfect game, as well as being a shutout victory. a perfect game is by definition also a no-hitter. 22 perfect games have been recorded in mlb history \ldots & ($s=0.47$) dictionary and translator for handheld new: sensagent is now available on your handheld \ldots{} with a sensagentbox, visitors to your site can access reliable information on over 5 million pages provided by sensagent.com \ldots \\
    \addlinespace
    Home-energy guide (25) & ($s=0.93$) \ldots{} energy efficiency starts with learning how to recognize a problem, making the right choices, and sometimes changing old habits. this book will help you become more aware of energy use in your home and in the world \ldots & ($s=0.50$) \ldots{} click for more about paul scheckel's important new book deus ex machina by paul freundlich click for more about paul scheckel's important new book \ldots{} copyright \copyright{} 2003-2012 by the public press: the page, all rights reserved \\
    \addlinespace
    Essay on digital magazines (34) & ($s=0.92$) \ldots{} for the past year i've been pushing the theory that the age of tablets will give print media one last bite at the apple -- and publishing companies that are able to make the transition could one day thrive again \ldots & ($s=0.46$) ford motor co \textbar{} 15.08 \textbar{} 0.44 \textbar{} 3.01\% \textbar{} general electric co \textbar{} 23.46 \textbar{} 0.19 \textbar{} 0.82\% \textbar{} glass employees speak openly on public concerns more \ldots \\
    \addlinespace
    Policy report on coal and carbon capture (25) & ($s=0.94$) \ldots{} although an emission performance standard that requires ccs systems for all new coal plants would pose a daunting technological and economic challenge, it will ultimately assure coal a secure and important role in the future u.s.\ energy mix \ldots & ($s=0.51$) \ldots{} print: [name] (women's issues, poverty, legal progress) [phone] or email@example.com print: [name] (spanish language and ethnic media) [phone] \ldots{} tv: \ldots{} radio: \ldots{} web: \ldots \\
    \addlinespace
    Article on the Ford F-Series (50) & ($s=0.91$) \ldots{} the f-series is a series of full-size pickup trucks from ford motor company which has been sold continuously for over six decades. the most popular variant of the f-series is the f-150 \ldots & ($s=0.46$) \ldots{} available to download on itunes including remixes of: tiesto, hard rock sofa \& deniz koyu http://smarturl.it/\ldots{} macklemore \& ryan lewis present the official music video for can't hold us \ldots \\
    \midrule
    \multicolumn{3}{@{}l}{\emph{What the filter misses}} \\
    \midrule
    Site policy (23) & ($s=0.95$) \ldots{} this site is controlled and operated by thesunsetchart.com from its offices within the state of california, united states of america. thesunsetchart.com makes no representation that materials in the site are appropriate or available for use in other locations \ldots & ($s=0.57$) some of our recent visitors\ldots{} bellevue, waunited states mountain view, caunited states seattle, waunited states \ldots{} copyright 2013 intelligenesis, all rights reserved \ldots \\
    \addlinespace
    Interview with a band (22) & ($s=0.91$) \ldots{} one of the things i admire most about cocore is the spirit of experimentation. can you describe your relationship to experimental music? who or what inspires you to push the limits? \ldots & ($s=0.79$) cocore is one of south korea's first and greatest indie rock bands. in the beginning of their career, they made a name for themselves by covering nirvana and penning their own grunge anthems \ldots \\
    \bottomrule
  \end{tabular}
\end{table}

\section{Full Objective}\label{app:objective}

The expressions for $\Lcross$ and $\Lself$ in Sections~2.3 and~2.4 are written for a single group. Over a batch of $B$ groups $g_b=(a_b,s_{b1},\dots,s_{bn})$, each is averaged over masked tokens rather than over groups. Let $S(y\mid x)=|M_H(y)|\,C(y\mid x)$ be the cross-entropy of Eq.~\ref{eq:rcross} summed over the masked positions of $y$, and $S_D(y)=|M_H(y)|\,D(y)$ the summed KL divergence of Eq.~\ref{eq:rself}.
\begin{itemize}
  \item \textbf{Cross-passage supervision.}
  \[
    \Lcross = \frac12\left[
      \frac{\sum_{b}\sum_{j=1}^{n} S(a_b\mid s_{bj})}{n\sum_{b}|M_H(a_b)|}
      + \frac{\sum_{b}\sum_{j=1}^{n} S(s_{bj}\mid a_b)}{\sum_{b}\sum_{j=1}^{n}|M_H(s_{bj})|}
    \right].
  \]
  Each passage has a single heavy view, so the $n$ predictions of an anchor share its masked positions; hence the factor $n$ in the first denominator.
  \item \textbf{Within-passage supervision.}
  \[
    \Lself = \frac12\left[
      \frac{\sum_{b} S_D(a_b)}{\sum_{b}|M_H(a_b)|}
      + \frac{\sum_{b} S_D(s_{b1})}{\sum_{b}|M_H(s_{b1})|}
    \right].
  \]
  \item \textbf{Relation to Section~\ref{sec:pred}.} If every heavy view masks the same number of tokens, both expressions reduce to the batch mean of the per-group terms in Section~\ref{sec:pred}. Passages have 64 to 128 tokens (Appendix~\ref{app:impl}), so the two differ only through shorter trailing passages.
  \item \textbf{MLM.} $\Lmlm$ is the cross-entropy of the MLM head on the lightly masked positions of $a$ and $s_1$.
  \item \textbf{Uniformity.} $\Lunif$ is computed over the anchors on each device.
\end{itemize}
All prediction losses are computed in fp32 from bf16 logits.
\section{Retention Regularizer}\label{app:retention}

In the continual setting, every method adds
\begin{equation}
  \Lret = \tfrac12\Big[\mathrm{KL}\big(\sigma(R_1R_2^{\top}/\tau_r)\,\big\|\,\sigma(E_1E_2^{\top}/\tau_r)\big) + \mathrm{KL}\big(\sigma(R_2R_1^{\top}/\tau_r)\,\big\|\,\sigma(E_2E_1^{\top}/\tau_r)\big)\Big],
  \label{eq:retention}
\end{equation}
where $E_1,E_2$ are normalized current embeddings of two sets of passages in the batch, $R_1,R_2$ are the corresponding embeddings of the frozen initial model, $\sigma$ is a row-wise softmax, and $\tau_r = 0.075$. Its weight decays as $\lambda_{\text{ret}}(t) = 40\,\big(1-\sqrt{(t-1)/(T-1)}\big)$ over the $T$ training steps, which protects the initial model early and lets the pre-training objective dominate late. Equation~\ref{eq:retention} is applied to every E5-initialized run and to no BERT-initialized run.
\begin{itemize}
  \item \textbf{\method{}.} $E_1,E_2$ are the mean-pooled embeddings of the lightly masked anchor and top sibling. $R_1,R_2$ are computed online by a frozen copy of E5 on the unmasked, prefixed passages. Similarity matrices are gathered across GPUs ($512\times512$).
  \item \textbf{Single-passage baselines.} The reference embeddings of the anchors are precomputed with the same frozen model on the same prefixed passages and stored in fp16. Both sets are the anchors of the batch.
  \item \textbf{Schedule.} For all methods, the weight halves after a quarter of training and reaches 0 at the last step.
\end{itemize}

\section{Baselines}\label{app:baselines}

\begin{table}[h]
  \centering
  \caption{\textbf{Baseline configurations.} Every baseline uses its official implementation, the same groups (the anchor passage; for Revela, the anchor with its siblings and anchors of other documents) and the same steps, batch size, learning-rate schedule and, in the continual setting, retention regularizer as \method{}.}
  \label{tab:baselines}
  \small
  \begin{tabular}{p{1.8cm}p{7.6cm}p{3.2cm}}
    \toprule
    Method & Native objective and settings & Passages per group \\
    \midrule
    RetroMAE & encoder MLM at 30\% masking; enhanced one-layer decoder reconstructs at 50\% & 1 (anchor) \\
    DupMAE & RetroMAE plus bag-of-words prediction from token states (weight 0.1) & 1 (anchor) \\
    SimLM & replaced-language modeling with a frozen ELECTRA-base generator; masks 30\% / 50\%; two-layer decoder; generator MLM weight 0.2 & 1 (anchor) \\
    BoW-DPR & encoder MLM at 30\%; bag-of-words prediction (coefficient 1.0, cosine-decayed to 0.1); head mask 15\% & 1 (anchor) \\
    Revela & next-token prediction by Qwen2.5-0.5B (LoRA, rank 256, $\alpha=64$) with in-batch attention weighted by the retriever's similarities; 128-token passages & 16 (anchor, 4 siblings, 11 anchors of other documents) \\
    \method{} & Eq.~\ref{eq:objective} & 5 (anchor + 4 siblings) \\
    \bottomrule
  \end{tabular}
\end{table}

\section{Extended Related Work}\label{app:related}

\paragraph{Contrastive pre-training for retrieval.} Without labels, positive pairs have been built with the inverse cloze task \citep{lee2019ict}, body--first-section and wiki-link pairs \citep{chang2020pretraining}, recurring spans \citep{ram2022spider} and random cropping \citep{izacard2022contriever}, and combined with corpus-aware pre-training on spans of the same document \citep{gao2022cocondenser}; large-scale weakly supervised pairs, filtered with a consistency check, underlie today's general-purpose embedders \citep{wang2022e5,li2023gte}. \method{} shares this contrastive form and changes the pair. Spider pairs passages of a document that share a recurring span, a lexical notion of relatedness; \method{} pairs an anchor that must represent its document with the passage most related to it by fused dense and lexical similarity, and every document contributes several such pairs. Our ablations show that this selection matters (\S\ref{sec:ablation}), and \method{} adds a signal on top of a model pre-trained contrastively on weakly supervised pairs rather than replacing such pre-training (\S\ref{sec:main}).

\paragraph{Bottleneck auto-encoding.} Weak decoders force information into a single vector in SEED \citep{lu2021seed}, Condenser \citep{gao2021condenser}, RetroMAE \citep{xiao2022retromae} and its duplex extension DupMAE \citep{liu2023retromae2}, SimLM with replaced-token detection \citep{wang2023simlm}, and LexMAE with a lexicon bottleneck \citep{shen2023lexmae}; BoW-DPR drops the decoder in favor of bag-of-words prediction \citep{ma2024bowdpr}. The closest work is CoT-MAE \citep{wu2023cotmae}, which reconstructs a span with the help of the embedding of its adjacent span. \method{} uses prediction differently: as content supervision alongside a contrastive objective, with targets drawn from related passages anywhere in the document.

\paragraph{Learning from related passages.} Revela \citep{cai2026revela} trains a retriever through language modeling: a language model predicts each passage of a batch while attending to the other passages, weighted by the retriever's similarities, so the retriever learns which passages help to predict a given one. \method{} shares the premise that related passages carry a training signal, but it selects the related passages within a document, makes the retriever's own embedding align with them and predict them, and needs no language model. Under our protocol, Revela does not improve E5 (\S\ref{sec:main}) and, from BERT, stays below SimLM (\S\ref{sec:results}).

\paragraph{Predictive self-distillation.} Predicting a model's own targets from a privileged view underlies self-supervised learning in vision and speech \citep{grill2020byol,caron2021dino,baevski2022data2vec,assran2023ijepa}. \method{} uses such targets in its within-passage term $\Lself$, where the prediction must pass through the passage's own embedding; its cross-passage term $\Lcross$ instead predicts the actual masked tokens of a related passage.

\paragraph{Conditioning on continuous vectors.} Prefix and prompt tuning condition a frozen or shared model on a few continuous vectors \citep{li2021prefix,lester2021prompt}. Gist tokens, AutoCompressor and ICAE compress contexts into such vectors for later generation \citep{mu2023gist,chevalier2023autocompressor,ge2024icae}. \method{} uses the same interface as a training-time bottleneck and asks the vectors to support prediction of other text rather than reconstruction of the compressed one.

\paragraph{Continual pre-training and forgetting.} Continued pre-training adapts models to new data \citep{gururangan2020dontstop} but can erase prior capabilities \citep{kirkpatrick2017ewc,li2016lwf}. Our retention regularizer follows similarity-preserving and relational distillation \citep{tung2019similarity,park2019relational}: it preserves the initial model's in-batch similarity structure and is annealed away during training.

\section{Additional Analysis}\label{app:analysis}

\paragraph{Examples.} Table~\ref{tab:examples} shows six queries, each from a different debate, among the 16 for which \method{} ranks the counter-argument first and both E5 and BoW-DPR rank it fifth or lower. In each, E5's first passage is about the subject of the query but does not answer it: it is another argument of the same debate (drugs, tourism in Tunisia, the International Criminal Court), the counter-argument to another point (assisted suicide), or an argument from a neighboring debate (EU foreign policy, aid for accepting migrants). \method{}'s first passage responds to the query's point directly, down to its example of an embassy in Djibouti or its tourist numbers for 2011. This is the kind of relation that training on related passages of a page rewards: what a passage is about, and which passages address it.

\begin{table}[h]
  \centering
  \caption{\textbf{ArguAna queries for which \method{} ranks the counter-argument first and E5 does not.} Six of the 16 queries for which \method{} ranks the counter-argument first and both E5 and BoW-DPR rank it fifth or lower, each from a different debate (topic in italics). Passages are abridged from the corpus (ellipses); ranks come from the full-precision re-encoding described in the text.}
  \label{tab:examples}
  \footnotesize
  \begin{tabular}{@{}p{1.75cm}p{11.6cm}@{}}
    \toprule
    \multicolumn{2}{@{}p{13.35cm}@{}}{\emph{Legalization of drugs.} Rank of the counter-argument: E5 9, BoW-DPR 15, \method{} 1.} \\
    \midrule
    Query & Legalisation reduces crime. The illegality of drugs fuels a huge amount of crime that could be eliminated if drugs were legalised. Price controls would mean that addicts would no longer have to steal to fund their habits \ldots \\
    E5, first & More people will take drugs if they are legal. Considering that drug use has so many negative consequences, it would be disastrous to have it increase \ldots \\
    \method{}, first & Drugs affect how people think, and they take away their ability to control their actions rationally, and so people on drugs are more likely to commit crimes \ldots\ ``Most drug crimes aren't committed by people trying to pay for drugs; they're committed by people on drugs.'' \\
    \midrule
    \multicolumn{2}{@{}p{13.35cm}@{}}{\emph{UK influence and EU membership.} Rank of the counter-argument: E5 10, BoW-DPR 19, \method{} 1.} \\
    \midrule
    Query & The EU is a force multiplier. The UK gets more bang for the buck as a result of being a member of the EU. It has representation in more countries \ldots\ For example the EU have representation in Djibouti \ldots \\
    E5, first & \emph{(from a debate on a common EU foreign policy)} The position of High Representative will be, and has been, largely powerless \ldots \\
    \method{}, first & It is undeniable that in some areas the EU is a force multiplier. But many of the issues it uses this leverage on are not areas of concern to a UK that has left the EU \ldots\ could the UK have an embassy in Djibouti? Certainly if it wished \ldots \\
    \midrule
\multicolumn{2}{@{}p{13.35cm}@{}}{\emph{Whether Tunisia should rely on tourism.} Rank of the counter-argument: E5 8, BoW-DPR 7, \method{} 1.} \\
    \midrule
    Query & Vulnerable to unrest. Relying on tourism ensures that the economy is at the mercy of unrest. The violence and break down in law and order following the Tunisian revolution resulted in a notable decrease in tourists \ldots\ from 6,487,000 in 2010 to 4,456,000 in 2011 \ldots \\
    E5, first & Investment. Tourism should be relied upon for economic growth as it attracts significant foreign investment. Tourism is the largest form of foreign currency income \ldots \\
    \method{}, first & The long term affects that instability have on industries such as tourism is overstated. \ldots\ tourism has recovered from the low point of 2011. In the first ten months of 2013 Tunisia attracted 5.5 million tourists \ldots \\
    \midrule
\multicolumn{2}{@{}p{13.35cm}@{}}{\emph{An enforcement arm for the ICC.} Rank of the counter-argument: E5 5, BoW-DPR 12, \method{} 1.} \\
    \midrule
    Query & ICC enforcement would create resentment. \ldots\ it would be a foreign force, that may or may not be seen as legitimate, attempting to arrest a native of that country. The result would be resentment in the community at the intrusion \ldots \\
    E5, first & An ICC Enforcement arm would bring in a higher proportion of defendants in to trial. Eight out of the thirty people indicted by the ICC \ldots\ are still alive and avoiding justice \ldots \\
    \method{}, first & \ldots\ Often a problem with arresting wanted international criminals is their support in the community \ldots\ The national government's enforcement arm may be as unwelcome as the ICC \ldots \\
    \midrule
\multicolumn{2}{@{}p{13.35cm}@{}}{\emph{Legalization of assisted suicide.} Rank of the counter-argument: E5 7, BoW-DPR 11, \method{} 1.} \\
    \midrule
    Query & It is vital that a doctor's role not be confused. The guiding principle of medical ethics is to do no harm \ldots\ assisted suicide also places an unreasonable burden on doctors \ldots \\
    E5, first & \emph{(the counter-argument to another point)} However, the idea that we should not kill is not absolute \ldots\ Concerning the notion that legalised voluntary euthanasia might lead to involuntary euthanasia being carried out, there is no evidence to suggest this \ldots \\
    \method{}, first & At the moment, doctors are often put into an impossible position. A good doctor will form close bonds with their patients \ldots\ when a patient \ldots\ expresses a strong desire to die, they are legally unable to help \ldots \\
    \midrule
\multicolumn{2}{@{}p{13.35cm}@{}}{\emph{An absolute right to asylum.} Rank of the counter-argument: E5 107, BoW-DPR 9, \method{} 1.} \\
    \midrule
    Query & The Whole System is broken. It is not clear that the system works at all. The majority of those who apply for asylum are working-age males, which implies that there is a strong economic angle \ldots \\
    E5, first & \emph{(from a debate on offering aid to developing countries that accept migrants)} Migrants will simply return to the countries they have been sent from \ldots \\
    \method{}, first & Much of the fear of the asylum system being used by economic migrants is simply media hysteria and xenophobia. The vast majority of asylum claims (in the UK around 75\%) are still rejected, which shows the system works \ldots \\
    \bottomrule
  \end{tabular}
\end{table}

\section{Additional Results}\label{app:results}

Table~\ref{tab:retrieval-metrics} complements Table~\ref{tab:main} with win/loss counts and retrieval metrics beyond nDCG@10, and Tables~\ref{tab:retrieval-scale}--\ref{tab:pertask-bert} give per-task scores.

\begin{table}[h]
  \centering
  \caption{\textbf{Average score, win/loss counts and retrieval metrics} (continual from E5; $\times 100$). $\Delta$ is the change in Avg.\ with respect to E5, computed before rounding, and W/L counts the tasks on which a model is better/worse than E5. Retrieval metrics are averaged over the ten retrieval tasks.}
  \label{tab:retrieval-metrics}
  \small
  \setlength{\tabcolsep}{4.5pt}
  \begin{tabular}{llccccccc}
    \toprule
    & & \multicolumn{3}{c}{MTEB(eng, v2), 41 tasks} & \multicolumn{4}{c}{Retrieval, 10 tasks} \\
    \cmidrule(lr){3-5}\cmidrule(lr){6-9}
    Corpus & Method & Avg. & $\Delta$ & W/L & nDCG@10 & MRR@10 & R@10 & R@100 \\
    \midrule
    -- & E5 (initial) & 53.06 & -- & -- & 43.26 & 51.25 & 43.31 & 65.25 \\
    \midrule
    \multirow{6}{*}{100K}
       & RetroMAE & 52.92 & $-0.14$ & 17/24 & 42.41 & 50.34 & 42.82 & 65.05 \\
       & DupMAE & 52.97 & $-0.09$ & 20/21 & 42.44 & 50.34 & 42.82 & 65.06 \\
       & SimLM & 51.01 & $-2.05$ & 18/23 & 36.93 & 46.65 & 36.26 & 56.56 \\
       & BoW-DPR & 53.14 & $+0.08$ & 23/18 & 43.11 & 51.95 & 42.92 & 64.16 \\
       & Revela & 53.02 & $-0.05$ & 23/18 & 42.99 & 50.63 & 43.07 & 65.28 \\
       & \method{} & \textbf{54.86} & $\mathbf{+1.80}$ & \textbf{34/7} & \textbf{45.78} & \textbf{53.88} & \textbf{46.05} & \textbf{67.50} \\
    \midrule
    \multirow{6}{*}{500K}
       & RetroMAE & 52.99 & $-0.07$ & 22/19 & 42.38 & 50.14 & 42.85 & 65.03 \\
       & DupMAE & 53.00 & $-0.06$ & 21/20 & 42.32 & 50.22 & 42.78 & 65.06 \\
       & SimLM & 53.32 & $+0.26$ & 26/15 & 42.01 & 49.54 & 41.26 & 63.95 \\
       & BoW-DPR & 53.44 & $+0.38$ & 27/14 & 43.38 & 51.09 & 42.91 & 65.14 \\
       & Revela & 52.89 & $-0.17$ & 14/27 & 42.66 & 50.16 & 42.83 & 65.17 \\
       & \method{} & \textbf{54.93} & $\mathbf{+1.87}$ & \textbf{36/5} & \textbf{45.67} & \textbf{53.78} & \textbf{46.00} & \textbf{67.42} \\
    \midrule
    \multirow{6}{*}{1M}
       & RetroMAE & 52.97 & $-0.09$ & 22/19 & 41.94 & 49.66 & 42.73 & 64.94 \\
       & DupMAE & 52.94 & $-0.12$ & 22/19 & 41.99 & 49.82 & 42.68 & 64.94 \\
       & SimLM & 51.36 & $-1.70$ & 11/30 & 37.93 & 46.33 & 37.49 & 59.19 \\
       & BoW-DPR & 53.02 & $-0.04$ & 23/18 & 42.36 & 50.30 & 41.81 & 64.39 \\
       & Revela & 52.80 & $-0.26$ & 16/25 & 42.47 & 49.51 & 42.89 & 65.31 \\
       & \method{} & \textbf{54.89} & $\mathbf{+1.83}$ & \textbf{35/6} & \textbf{45.33} & \textbf{53.53} & \textbf{45.81} & \textbf{67.56} \\
    \bottomrule
  \end{tabular}
\end{table}

\begin{table}[h]
  \centering
  \caption{\textbf{Retrieval (nDCG@10, $\times 100$) per task across scales} (continual from E5).}
  \label{tab:retrieval-scale}
  \small
  \begin{tabular}{lcccc}
    \toprule
    Task & E5 (initial) & \method{} 100K & \method{} 500K & \method{} 1M \\
    \midrule
    ArguAna & 40.36 & 46.12 & 45.77 & 46.04 \\
    CQADupstack-Gaming & 53.45 & 56.06 & 56.12 & 56.28 \\
    CQADupstack-Unix & 35.26 & 37.00 & 37.00 & 36.50 \\
    Climate-FEVER-HN & 16.48 & 20.06 & 20.47 & 20.94 \\
    FEVER-HN & 65.58 & 69.65 & 69.83 & 69.78 \\
    FiQA-2018 & 40.00 & 38.87 & 38.30 & 37.71 \\
    HotpotQA-HN & 56.18 & 59.19 & 59.19 & 58.87 \\
    SCIDOCS & 21.35 & 21.84 & 21.77 & 21.80 \\
    TREC-COVID & 64.39 & 62.96 & 62.52 & 60.90 \\
    Touch\'e-2020 v3 & 39.51 & 46.08 & 45.72 & 44.51 \\
    \midrule
    Average & 43.26 & 45.78 & 45.67 & 45.33 \\
    \bottomrule
  \end{tabular}
\end{table}

\begin{table}[p]
  \centering
  \caption{\textbf{Per-task main scores at 100K} (continual from E5; $\times 100$).}
  \label{tab:pertask}
  \scriptsize
  \setlength{\tabcolsep}{4pt}
  \begin{tabular}{llccccccc}
    \toprule
    Type & Task & E5 & RetroMAE & DupMAE & SimLM & BoW-DPR & Revela & \method{} \\
    \midrule
    Class. & AmazonCounterfactual & 64.22 & 63.82 & 63.85 & 67.08 & 65.76 & 64.54 & 66.34 \\
     & Banking77 & 75.81 & 75.71 & 75.70 & 76.64 & 76.48 & 75.83 & 77.07 \\
     & Imdb & 66.15 & 68.22 & 68.08 & 60.17 & 61.01 & 67.24 & 64.66 \\
     & MTOPDomain & 62.54 & 62.57 & 62.64 & 63.37 & 63.41 & 62.55 & 63.24 \\
     & MassiveIntent & 27.35 & 27.33 & 27.31 & 29.62 & 28.89 & 27.49 & 28.83 \\
     & MassiveScenario & 34.38 & 34.35 & 34.32 & 36.46 & 35.91 & 34.44 & 35.62 \\
     & ToxicConversations & 63.18 & 63.32 & 63.35 & 64.25 & 65.00 & 64.03 & 65.35 \\
     & TweetSentimentExtraction & 52.11 & 52.35 & 52.37 & 52.17 & 52.55 & 52.68 & 52.82 \\
    \midrule
    Clust. & ArXivHierarchicalP2P & 61.97 & 61.64 & 61.37 & 55.14 & 59.79 & 61.38 & 61.61 \\
     & ArXivHierarchicalS2S & 58.83 & 58.52 & 58.52 & 57.33 & 58.34 & 58.94 & 60.19 \\
     & BiorxivP2P.v2 & 41.58 & 41.49 & 41.22 & 28.21 & 38.68 & 41.89 & 41.26 \\
     & MedrxivP2P.v2 & 37.13 & 37.12 & 36.87 & 29.69 & 33.90 & 36.91 & 37.57 \\
     & MedrxivS2S.v2 & 35.67 & 35.03 & 35.18 & 33.34 & 34.42 & 35.68 & 35.68 \\
     & StackExchange.v2 & 55.04 & 55.10 & 55.39 & 47.32 & 50.76 & 55.23 & 57.11 \\
     & StackExchangeP2P.v2 & 39.89 & 39.71 & 40.02 & 34.36 & 38.87 & 39.54 & 42.46 \\
     & TwentyNewsgroups.v2 & 38.37 & 36.96 & 38.59 & 31.53 & 33.86 & 37.35 & 38.07 \\
    \midrule
    Pair & SprintDuplicateQuestions & 90.77 & 90.95 & 90.95 & 93.94 & 92.58 & 90.91 & 91.25 \\
     & TwitterSemEval2015 & 62.76 & 62.78 & 62.78 & 63.55 & 63.24 & 62.81 & 63.77 \\
     & TwitterURLCorpus & 84.03 & 84.00 & 84.01 & 84.51 & 84.40 & 83.77 & 84.48 \\
    \midrule
    Rerank & AskUbuntuDupQuestions & 57.99 & 57.68 & 57.66 & 57.42 & 57.62 & 57.96 & 57.97 \\
     & MindSmallReranking & 31.62 & 31.73 & 31.72 & 31.16 & 31.58 & 31.97 & 31.70 \\
    \midrule
    Retr. & ArguAna & 40.36 & 39.48 & 39.56 & 22.50 & 38.89 & 39.44 & 46.12 \\
     & CQADupstack-Gaming & 53.45 & 53.36 & 53.45 & 50.76 & 51.75 & 53.28 & 56.06 \\
     & CQADupstack-Unix & 35.26 & 34.99 & 35.12 & 33.62 & 34.74 & 35.55 & 37.00 \\
     & Climate-FEVER-HN & 16.48 & 15.72 & 15.68 & 15.51 & 16.90 & 14.55 & 20.06 \\
     & FEVER-HN & 65.58 & 63.79 & 63.65 & 49.52 & 63.11 & 64.92 & 69.65 \\
     & FiQA-2018 & 40.00 & 39.80 & 39.57 & 27.88 & 39.46 & 39.50 & 38.87 \\
     & HotpotQA-HN & 56.18 & 55.80 & 55.84 & 60.79 & 59.52 & 57.02 & 59.19 \\
     & SCIDOCS & 21.35 & 21.38 & 21.37 & 19.49 & 21.13 & 21.48 & 21.84 \\
     & TREC-COVID & 64.39 & 60.95 & 61.16 & 57.39 & 65.78 & 64.72 & 62.96 \\
     & Touch\'e-2020 v3 & 39.51 & 38.80 & 38.97 & 31.79 & 39.87 & 39.42 & 46.08 \\
    \midrule
    STS & BIOSSES & 71.97 & 73.21 & 73.33 & 82.63 & 79.14 & 71.49 & 81.35 \\
     & SICK-R & 68.68 & 68.86 & 68.89 & 67.97 & 69.02 & 68.46 & 70.04 \\
     & STS12 & 57.84 & 57.80 & 57.84 & 57.54 & 58.04 & 57.93 & 59.50 \\
     & STS13 & 73.46 & 73.92 & 73.90 & 74.97 & 74.37 & 73.17 & 73.71 \\
     & STS14 & 64.04 & 64.28 & 64.30 & 65.68 & 65.08 & 63.73 & 65.99 \\
     & STS15 & 75.44 & 76.02 & 76.04 & 79.98 & 78.06 & 74.87 & 78.43 \\
     & STS17 & 47.57 & 47.62 & 47.61 & 48.47 & 48.09 & 47.71 & 50.50 \\
     & STS22.v2 & 40.74 & 40.61 & 40.67 & 43.84 & 45.42 & 40.15 & 47.58 \\
     & STSBenchmark & 70.51 & 70.95 & 70.96 & 75.81 & 73.78 & 70.99 & 74.53 \\
    \midrule
    Summ. & SummEval.v2 & 31.46 & 32.19 & 32.13 & 28.10 & 29.57 & 32.13 & 32.87 \\
    \bottomrule
  \end{tabular}
\end{table}

\begin{table}[p]
  \centering
  \caption{\textbf{Per-task main scores when pre-training from BERT} (100K groups, six epochs; $\times 100$; \method{}: mean of three seeds).}
  \label{tab:pertask-bert}
  \scriptsize
  \setlength{\tabcolsep}{4pt}
  \begin{tabular}{llcccccc}
    \toprule
    Type & Task & RetroMAE & DupMAE & SimLM & BoW-DPR & Revela & \method{} \\
    \midrule
    Class. & AmazonCounterfactual & 69.80 & 69.45 & 67.92 & 69.77 & 63.41 & 59.99 \\
     & Banking77 & 56.93 & 53.69 & 60.24 & 59.73 & 60.37 & 65.69 \\
     & Imdb & 60.84 & 59.13 & 65.00 & 62.73 & 63.69 & 61.56 \\
     & MTOPDomain & 54.67 & 53.28 & 52.57 & 53.88 & 52.19 & 54.30 \\
     & MassiveIntent & 24.01 & 22.75 & 21.32 & 22.57 & 21.98 & 24.11 \\
     & MassiveScenario & 29.21 & 28.18 & 26.82 & 27.80 & 27.73 & 29.64 \\
     & ToxicConversations & 66.22 & 66.48 & 68.94 & 70.36 & 67.49 & 65.44 \\
     & TweetSentimentExtraction & 49.75 & 48.34 & 53.37 & 52.87 & 53.15 & 49.94 \\
    \midrule
    Clust. & ArXivHierarchicalP2P & 54.52 & 54.09 & 56.92 & 56.07 & 55.96 & 57.44 \\
     & ArXivHierarchicalS2S & 52.67 & 51.42 & 54.36 & 53.62 & 52.62 & 54.02 \\
     & BiorxivP2P.v2 & 34.43 & 32.26 & 35.00 & 33.28 & 33.22 & 36.86 \\
     & MedrxivP2P.v2 & 31.66 & 30.52 & 32.63 & 31.36 & 31.84 & 34.87 \\
     & MedrxivS2S.v2 & 27.42 & 26.89 & 28.33 & 28.10 & 28.64 & 32.18 \\
     & StackExchange.v2 & 42.62 & 42.10 & 50.35 & 48.42 & 53.63 & 57.91 \\
     & StackExchangeP2P.v2 & 33.09 & 32.44 & 34.97 & 33.70 & 34.19 & 39.05 \\
     & TwentyNewsgroups.v2 & 19.23 & 16.59 & 28.76 & 28.30 & 32.16 & 41.97 \\
    \midrule
    Pair & SprintDuplicateQuestions & 76.48 & 73.13 & 74.55 & 72.23 & 73.26 & 90.27 \\
     & TwitterSemEval2015 & 60.04 & 60.00 & 60.69 & 60.91 & 58.77 & 57.88 \\
     & TwitterURLCorpus & 80.31 & 79.00 & 79.41 & 80.56 & 77.60 & 80.43 \\
    \midrule
    Rerank & AskUbuntuDupQuestions & 49.66 & 49.71 & 48.87 & 50.31 & 49.54 & 51.50 \\
     & MindSmallReranking & 27.54 & 26.56 & 28.33 & 28.19 & 29.31 & 30.04 \\
    \midrule
    Retr. & ArguAna & 33.49 & 31.61 & 41.29 & 35.58 & 28.49 & 40.69 \\
     & CQADupstack-Gaming & 25.45 & 25.14 & 29.51 & 28.39 & 25.75 & 36.05 \\
     & CQADupstack-Unix & 14.84 & 13.99 & 16.13 & 14.43 & 13.58 & 20.32 \\
     & Climate-FEVER-HN & 8.90 & 8.44 & 8.97 & 8.99 & 7.17 & 14.54 \\
     & FEVER-HN & 20.59 & 18.00 & 19.30 & 18.08 & 14.17 & 45.00 \\
     & FiQA-2018 & 7.21 & 5.90 & 7.64 & 6.62 & 6.44 & 16.19 \\
     & HotpotQA-HN & 28.66 & 27.95 & 28.84 & 23.96 & 30.52 & 43.03 \\
     & SCIDOCS & 8.81 & 8.31 & 8.82 & 7.65 & 6.96 & 13.27 \\
     & TREC-COVID & 35.11 & 35.15 & 29.24 & 30.57 & 27.11 & 31.76 \\
     & Touch\'e-2020 v3 & 13.42 & 14.42 & 17.10 & 16.10 & 14.30 & 31.29 \\
    \midrule
    STS & BIOSSES & 64.31 & 64.82 & 67.90 & 61.37 & 65.10 & 74.24 \\
     & SICK-R & 61.51 & 61.21 & 61.48 & 64.55 & 66.08 & 65.05 \\
     & STS12 & 46.87 & 43.42 & 46.41 & 47.46 & 48.59 & 53.37 \\
     & STS13 & 63.88 & 61.13 & 64.38 & 68.72 & 70.58 & 72.99 \\
     & STS14 & 58.22 & 55.69 & 56.54 & 60.26 & 56.04 & 63.46 \\
     & STS15 & 72.71 & 70.46 & 68.91 & 74.46 & 70.29 & 74.43 \\
     & STS17 & 35.15 & 33.93 & 27.34 & 32.40 & 34.85 & 35.76 \\
     & STS22.v2 & 27.02 & 26.17 & 29.64 & 24.05 & 28.95 & 33.46 \\
     & STSBenchmark & 59.70 & 57.12 & 61.72 & 64.51 & 64.83 & 65.73 \\
    \midrule
    Summ. & SummEval.v2 & 33.28 & 33.77 & 31.38 & 32.29 & 31.35 & 26.23 \\
    \bottomrule
  \end{tabular}
\end{table}


\end{document}